\documentclass[preprint]{elsarticle}

\usepackage{amssymb}

\journal{Journal of Logical and Algebraic Methods in Programming}

\usepackage{hyperref}
\usepackage{amsthm}
\usepackage{stmaryrd}
\usepackage{mathtools}
\usepackage[all,cmtip]{xy}
\xyoption{rotate}
\usepackage{etex}
\usepackage{tikz}
\usepackage{pgfplots}
\usepackage{verbatim}
\usepackage{multicol}
\usepackage{thmtools}
\usepackage{proof}
\usepackage[mathcal]{euscript}
\newcommand{\xykarrow}{ {\SelectTips{cm}{}\scalebox{0.75}{\object@{|}} } }

\newcommand{\pullbackcorner}[1][dr]{
  \save*!/#1-1.2pc/#1:(-1,1)@^{|-}\restore
}

\newcommand{\karrow}{\mathrel{\mathmakebox[\widthof{$\xrightarrow{\rule{1.45ex}{0ex}}$}]{\xrightarrow{\rule{1.45ex}{0ex}}\hspace*{-2.4ex}{\mapstochar}\hspace*{1.8ex}}}}

\newcommand{\eqdef}{\widehat{=} \> \>}
\newcommand{\conc}{\> {\small \texttt{+} \hspace{-0.1cm} \texttt{+}} \>}
\newcommand{\conce}{\emph{ {\small \texttt{+} \hspace{-0.315cm} \texttt{+}}}}

\newcommand{\longerarrow}{\xrightarrow{\hspace*{1cm}}}
\newcommand{\klH}{\topo_{\MH}}
\newcommand{\topo}{{\bf Top}}
\def\justl#1#2{\\
        &#1& \rule{2em}{0pt}  \{
        \mbox{  \rule[-.7em]{0pt}{1.8em} {\footnotesize #2} \} } \\ && }
\newcommand{\MH}{\mathcal{H}}

\newcommand{\Rz}{\textsf{T}}
\newcommand{\Dur}{\textsf{D}}
\makeatletter
\renewcommand{\boxed}[1]{\text{\fboxsep=.1em\fbox{\m@th$\displaystyle#1$}}}
\makeatother

\def\psync#1{\llparenthesis #1 \rrparenthesis}

\def\comp{\mathbin{\boldsymbol{\cdot}}}
\def\kcomp{ \> \bullet \> }
\def\pv#1#2{\langle #1 \rangle #2}

\def\hmul#1#2#3 { #1 \> \lhd \> #2 \> \rhd \> #3 }
\DeclareMathOperator{\img}{img}
\declaretheorem[style=mystyle,name=Example]{Mexample}
\declaretheorem[style=mystyle,name=Definition]{Mdefinition}

\declaretheorem[style=mystyle,name=Theorem]{Mtheorem}
\declaretheorem[style=mystyle,name=Lemma]{Mlemma}

\declaretheorem[style=mystyle,name=Corollary]{Mcorollary}
\def\const#1{\underline{#1}}
\def\typ#1{\stackrel{#1}{\longerarrow}}
\def\cond#1#2#3 { #1 \> \lhd \> #2 \> \rhd \> #3 }
\def\corn#1 { \widehat{#1}   }
\def\pulb#1{\pv{\pv{#1}}}

\newcommand{\ie}{\emph{i.e.}}

\begin{document}

\begin{frontmatter}



\allowdisplaybreaks[3]

\title{Continuity as a computational effect}

\author[1]{Renato Neves}
\author[1]{Luis S. Barbosa}
\author[2]{Dirk Hofmann}
\author[2]{Manuel A. Martins}
\address[1]{ INESC TEC (HASLab) \& Universidade do Minho, Portugal\\
  \texttt{rjneves@inescporto.pt,lsb@di.uminho.pt}}
\address[2]{CIDMA - Dep. of Mathematics, Universidade de Aveiro, Portugal \\
  \texttt{\{dirk,martins\}@ua.pt}}

\begin{abstract}

The original purpose of component-based development was to provide
techniques to master complex software, through composition, reuse and
parametrisation. However, such systems are rapidly moving towards a
level in which software becomes prevalently intertwined with
(continuous) physical processes. A possible way to accommodate the
latter in component calculi relies on a suitable encoding of
continuous behaviour as (yet another) computational effect.

This paper introduces such an encoding through a monad which, in the
compositional development of hybrid systems, may play a role
similar to the one played by $1+$, powerset, and distribution
monads in the characterisation of partial, nondeterministic and
probabilistic components, respectively. This monad and its Kleisli
category provide a universe in which the effects of continuity over
(different forms of) composition can be suitably studied.

\end{abstract}


\end{frontmatter}

\section{Introduction}

\subsection{Motivation and objectives.}

Component-based software development is often explained through a
visual metaphor: a palette of computational units, and a blank canvas
in which they are dropped and interconnected by drawing wires
abstracting different composition and synchronisation mechanisms.
More and more, however, components are not limited to traditional
information processing units, but encapsulate some form of interaction
with physical processes.  The resulting systems, referred to as
\emph{hybrid} \cite{tabuada09,alur2015}, exhibit a complex dynamics in which
computations, coordination, and physical processes interact, become
mutually constrained, and cooperate to achieve specific goals.

One generic way of looking at components, proposed in
\cite{Barbosa03}, emphasises an \emph{observational} semantics,
through a signature of observers and methods, that makes them
amenable to a \emph{coalgebraic} characterisation as (generalisations
of) abstract Mealy machines. The resulting calculus is parametric on
whatever behavioural model underlies a component specification. This
captures, for example, partial, nondeterministic or probabilistic
behaviour of a component's dynamics by encoding such behavioural
effects as \emph{strong monads} \cite{kock:1972} --- a pervasive
mathematical structure with surprising applications in different areas
of Computer Science (see \emph{e.g.}, 
\cite{moggi:1991,wadler:1995,Marlow11,Doberkat:2009,Hasuo06}).

Indeed, each monad captures a specific type of behaviour, which is
then reflected in the corresponding component calculus.  For example,
\emph{maybe} monad $(1+)$ introduces \emph{partial} components;
the \emph{powerset} $(\mathcal{P})$ monad \emph{nondeterministic}
ones; and \emph{distribution} monad $(\mathcal{D})$ brings
(discrete) \emph{probabilistic} evolution into the scene.  Can
\emph{continuous behaviour}, prevalent in hybrid systems and control
theory, be encoded in a similar way, as (yet another) computational
effect? Such is the question addressed in this paper.

Monads first came in contact to Computer Science in the 80's, when
E. Moggi proposed their use to structure the denotational semantics of
programming languages \cite{moggi:1989,moggi:1991}.
Later the concept was introduced in programming practice by P. Wadler
\cite{wadler:1995}, leading to a rigorous style of combining purely
functional programs that mimic impure (side-)effects.  The key idea is
that monads encode in abstract terms several kinds of computational
effects, such as exceptions, state updating, nondeterminism or
continuations. Such effects are represented by a type constructor
$\mathcal{T}$ (an endofunctor over a suitable category) so that
computations producing values of type $O$ are regarded as terms of
type $\mathcal{T} O$. In this way \emph{values} and
\emph{computations} are explicitly distinguished and programs can be
thought of as arrows $I \rightarrow \mathcal{T} O$ representing the
computation of values of type $O$ from values of type $I$, while
producing some effect described by $\mathcal{T}$. Or, putting it in a
different way, output values are encapsulated (or embedded) in the
effect specified by $\mathcal{T}$.  A monad comes equipped with an
identity and an associative multiplication which, from a computational
point of view, builds a (trivial) computation from a value, and
flattens nested effects, respectively. Furthermore,
if $\mathcal{T}$ is \emph{strong} \cite{wadler:1995} additional
machinery is available to distribute the computations' effect over
context.  The monad structure allows program composition by handling
the underlying computational effect through functor $\mathcal{T}$ and
the flattening operation. Actually, each monad gives rise to a so
called Kleisli category in which one may study the effects of the
behavioural type (as specified by the monad) over different forms of
composition; ultimately, this leads to rich component calculi (as
discussed in \cite{Barbosa03}).

The current paper introduces a (strong) monad $\mathcal{H}$ that
subsumes the typical continuous behaviour of dynamical, and hybrid
systems. Intuitively, the type effect of $\MH$
(\ie, the underlying endofunctor) represents the (continuous)
\emph{evolution} over time of some value in $O$; the identity
defines a trivial evolution (\ie, with duration zero), and the
flattening operation allows the control of an evolution to be passed
along different systems.

Moreover, the paper explores the corresponding Kleisli category as the
mathematical space in which the underlying (continuous) behaviour can
be isolated and its effect over different forms of composition
suitably studied. As we will see in the sequel,
such a category gives rise to several forms of composition operators
(\emph{e.g.},  sequential, parallel execution), \emph{wiring}
mechanisms, and \emph{synchronisation} techniques.
Again this parallels the role that the categories of partial
functions, relations and stochastic matrices have as reasoning
universes for component composition under the behavioural model
provided, respectively, by monads $1+$, $\mathcal{P}$ and
$\mathcal{D}$ \cite{BO06,JNO14}.  Similarly, this work paves the way
to the development of a coalgebraic calculus of \emph{hybrid
components} in the spirit of \cite{Barbosa03}.

\subsection{A tribute to Jos\'{e} Nuno Oliveira.}

The idea of regarding \emph{continuity} as a computational effect, or
more rigorously, a \emph{physical} one, entailing a suitable
notion of composition and a reasoning universe, in the form of a
Kleisli category, owes much to the way Jos\'{e} helped us to approach
computational phenomena.

Building on the role of monads in functional programming and program
calculi, as monadic inductive and coinductive schemes
\cite{mj95,par98}, Jos\'{e} introduced us to monads both as a powerful
structuring mechanism and a source of equally powerful genericity.  An
obsession for patterns and a sharp intuition for generic, conceptually
reusable structures remain, after all, the hallmark of his
illuminating, socratic teaching.

In the late 90's, Jos\'{e} supervised the PhD work of the second
author on the coalgebraic calculus of state-based components 
mentioned above \cite{Barbosa03}. This emerged from the conjunction
of two key ideas; first, that a `black-box' characterisation of
software components favoured an \emph{observational}, essentially
coalgebraic, semantics; second, that the envisaged calculus had to
be \emph{generic}, in the sense that it should not depend on a
particular notion of component behaviour.  Monads, actually strong
monads, were quickly identified as a source of such a genericity, the
whole work boiling down to a calculus of \emph{monadic} Mealy
machines. 
Software components were thus studied as coalgebras
(in a suitable category) typed as
\begin{flalign*}
  S \longrightarrow \mathcal{T}(S \times O)^I
\end{flalign*}
where $S$ represents the (internal) state space, and $I$, $O$ are
respectively the input and output spaces. $\mathcal{T}$ is a strong
monad that captures the intended behavioural effect.

Being generic entailed the need for an equally generic reasoning
framework. By then, the adoption of a \emph{pointfree}, essentially
equational, calculational proof style, thus avoiding the somehow more
standard coinductive proofs through the explicit construction of
bisimulations, was understood as the price to be paid for genericity,
as component laws were to be verified without fixing the working monad
completely. Generic proofs performed in this style are clear and easy
to follow, even if often long due to the systematic recording of
almost all elementary steps.

For Jos\'{e}, however, the way proofs are written is not a
technicality. Proofs, as he taught us every day, are basically honest
explanations, bearing evidence in a fixed formal context, and
therefore must be conveyed in a crisp, clear, easily reproducible
style, letting the underlying structure to emerge and helping to build
the correct intuitions. Years later, in the context of a joint
research project \cite{FMBB09}, Jos\'{e} championed the use of
calculational, pointfree reasoning as a way of reinvigorating the role
of proof in elementary mathematical education.  The pointfree style
adopted in many proofs of this paper is also intended as a tribute to
this view.

For Jos\'{e} being generic does not mean to seek refuge in
some sort of formal ivory tower, of stylised constructions polished
ahead of any meaningful intuition. This explains why, being a devoted
functional programmer, who resorts to Haskell as a pocket calculator,
Jos\'{e} soon started to focus his attention on the rich universes of
specific monadic computations --- their Kleisli categories. If pure
functions are computations for the identity monad, relations and
matrices play a similar role in such richer contexts.  To be added, of
course, and in a very concrete way, to the relevant calculator. His
systematic, calculational, `syntax-driven' work on relation algebra
\cite{oliveira09,MO12a}, as a framework for nondeterministic
computations, and linear algebra \cite{Ol12a,MO12}, for probabilistic
ones, was responsible for a fresh understanding of the Kleisli
categories of two fundamental monads, and lead to a number of new
results and simpler, elegant renderings of old ones. Having introduced
a monad for continuity, this paper initiates the unravelling of the
corresponding Kleisli category, as the reasoning universe for
continuous processes, thus, and once again, pursuing a path Jos\'{e}
will certainly cheer.

\subsection{Document structure}

After a brief detour on preliminaries and notation in Section
\ref{sec:pre}, the \emph{continuous evolution} monad ($\mathcal{H}$)
is introduced in Section \ref{sec:mon}. In Section \ref{sec:kl}, we
explore the corresponding Kleisli category: as we will see, its arrows
define \emph{continuous systems} $I \rightarrow \MH O$ (technically,
preliminary versions of dynamical, and hybrid systems) and (Kleisli)
composition makes possible for a component to execute after another,
starting its evolution when the preceding one finishes its own. In
Section \ref{sec:Adj}, we take advantage of the so called Kleisli
adjunction to define wiring mechanisms and characterise
(co)limits. The latter give rise to new forms of component composition
and corresponding laws. In order to add synchronisation techniques to
our (monadic) framework, Section \ref{sec:Sync} provides extra
structure to the underlying functor of monad $\MH$. After this we
suggest a feedback operator.  In Section \ref{sec:Str}, we show that
monad $\MH$ is strong; this brings us closer to hybrid systems as
coalgebraic components (in the spirit of \cite{Barbosa03}) whose
behavioural effect is captured by $\MH$. Formally, coalgebras typed as
\begin{flalign*} 
  S \longrightarrow \MH (S \times O)^I.
\end{flalign*} 
Finally, Section \ref{sec:con} discusses related work, provides
possible research directions, and presents concluding remarks. 

In order to illustrate the developments of the ensuing sections, a
number of classical examples of continuous and hybrid systems will be
explored under the light of the framework reported in this paper.



\section{Preliminaries}
\label{sec:pre}

\subsection{Continuous systems}

Technically, we qualify as \emph{continuous} a system whose output,
for any given input, is a (continuous) evolution over time;
\emph{i.e.}, an arrow typed as
\begin{flalign*} I \longrightarrow \coprod \limits_{d \in [0,\infty]}
O^{T_d}
\end{flalign*} where $I$, $O$ are, respectively, input and output
spaces, $O^{T_d}$ the space of continuous functions $T_d \rightarrow
O$ (the \emph{evolutions}), and $T_d$ stands for $\{ r \in
\mathbb{R}_{\geq 0} \> | \> r \leq d \}$.  Actually, this definition
includes the family of \emph{continuous dynamical systems} that
interpret the non-negative reals (\ie, $\mathbb{R}_{\geq 0}$, here
denoted by letter $\Rz$) as a time domain (\emph{cf.} \cite{introdyn,
  introdyn2}).  Formally, the latter are characterised as functions,

\begin{center}
    $\infer={\lambda \Phi : X \rightarrow X^{\Rz}}
    {\Phi : X \times \Rz \rightarrow X  }$  
\end{center}

\noindent
such that for any $t \in \Rz$, $x \in X$
\begin{flalign}
  & \Phi \> (x, 0) = x \label{predyn} \\
  & \Phi \> (x, t_1 + t_2) = \Phi (\Phi(x,t_1), \> t_2)
\end{flalign}

\noindent
From a monadic perspective, continuous dynamical systems (in the
form $\lambda \Phi : X \rightarrow X^{\Rz}$) may be
seen as programs whose behavioural effect subsumes some form of
continuous evolution over time. Indeed, as we will see later in the
paper, such systems are part of a broader family of arrows that live
in the Kleisli category of monad $\MH$ ($\klH$). In general, 
law \ref{predyn} will be an important part in the characterisation of
Kleisli composition.
We will also see that the traditional view of hybrid systems --
as a family of dynamical (or continuous) systems indexed by a (discrete)
state space -- coincides with ours; and, moreover, that such systems also
live in $\klH$ (due to the machinery that makes $\MH$ strong).

\subsection{Notation}

The key role that continuity takes in this work, suggests the category
$\topo$ of topological spaces and continuous functions as a suitable
working environment for developing the envisaged results. 

In the sequel, whenever the context is clear, a topological space will
be denoted by its underlying set. Topological spaces $X \times Y$, $X
+ Y$ correspond to the canonical product and coproduct of $X,Y$,
respectively. Also, for any $X \subseteq Y$, assume that $X$ has the
subspace topology induced by $Y$.  Finally, whenever $Y$ is
\emph{core-compact} (\emph{cf.} \cite{escardo2001}), space $X^Y$ has
the exponential topology.

Category $\topo$ is (co)complete; this allows to take advantage of
isomorphisms $\alpha : (X \times Y) \times Z \cong X \times (Y \times
Z)$, and $sw : X \times Y \cong Y \times X$. $\topo$ also provides a
set of useful rules for showing continuity; Figure \ref{fig:cont} sums
up the ones used in the paper.  In rule $( \> \lambda \> )$, $Y$ must
be core-compact so that the evaluation function $ev : X^Y \times Y
\rightarrow X$ is well defined (\emph{cf.} \cite{escardo2001}).

\begin{figure}
\begin{center}
\begin{tabular}{c  c}

  $\infer[(\> \comp \> )]{g \comp f : X \rightarrow Z}
    {f : X \rightarrow Y, g : Y \rightarrow Z}$  

  & 

  $\infer=[( \> \lambda \> )]{ \lambda f : X \rightarrow Z^{Y}}
    {f : X \times Y \rightarrow Z}$  
    
  \\ & \\

  $\infer=[( \> \times \> ) ]{\pv{f,g} : X \rightarrow Y_1 \times Y_2}
  {f : X \rightarrow Y_1, g : X \rightarrow Y_2}$

  &
  
  $\infer=[( \>  + \> ) ]{[f,g] : X_1 + X_2 \rightarrow Y}
  {f : X_1 \rightarrow Y, g : X_2 \rightarrow Y}$

  \\ & \\ 

   $\infer[( \> \downarrow_l \>  )]{f_A : A \rightarrow Y}
   {f : X \rightarrow Y, A \subseteq X}$

   &

   $\infer[( \> \downarrow_r \> )]{f^B : X \rightarrow B}
   {f : X \rightarrow Y, \img \> f \subseteq B}$ \\
     with $f_A \> = f \comp \iota \> $ 
     (for $\iota : A \hookrightarrow X$) & 
     with $\iota \comp f^B = f $ 
     (for $\iota : B \hookrightarrow Y$)\\ & 

\end{tabular}

\noindent 
\end{center}
\caption{Continuity rules in $\topo$.}
\label{fig:cont}
\end{figure}

Universal arrows $X \rightarrow 1$ to the final object in
$\topo$ are denoted by $!$, and a function constantly
yielding a value $x$ by $\underline{x}$. Given two functions $f,g : X
\rightarrow Y$, and a predicate $p$, we introduce a conditional
expression $f \> \lhd \> p \> \rhd \> g : X \rightarrow Y$, defined by,
\begin{flalign*}
& (\hmul{f}{p}{g}) \> x \> \> \eqdef
  \begin{cases}
    f \> x & \mbox{ if } p \> x \\
    g \> x & \mbox{ otherwise }
  \end{cases}
\end{flalign*}

\noindent
Whenever found relevant, and no ambiguities arise, we will denote 
expression $(\hmul{f}{p}{g}) \> x$ by $(\hmul{f \> x }{p \> x}{g \> x})$.
The continuous functions \emph{minimum} $\curlywedge : \Rz \times
[0,\infty] \rightarrow \Rz$ and \emph{truncated subtraction}
$\circleddash : \Rz \times [0,\infty] \rightarrow \Rz$ play a key role
in some proofs.  They are defined by the following equations
\begin{flalign*}
   & \curlywedge \> \eqdef 
   \hmul{\pi_1}{(\leq)}{\pi_2} \\
   & \circleddash  \> \eqdef \hmul{(-)}{(>)}{\underline{0}}  
\end{flalign*}
where $\leq,>$ are the usual ordering relations over the reals with
infinity.

As usual, functions $\pi_1 : X \times Y \rightarrow X$,
$\pi_2 : X \times Y \rightarrow Y$ correspond to the projections
associated with any binary product, and $i_1 : X \rightarrow X + Y$,
$i_2 : Y \rightarrow X + Y$ the coprojections associated with any
binary coproduct. Moreover, symbol $\star$ is used to denote the
element of a singleton set, and $| \mathbf{C} |$ to represent the
class of objects of a category $\mathbf{C}$. Finally, to avoid a
burdened notation, we will often drop the subscript in a
component of a natural transformation.

\section{The continuous evolution monad}
\label{sec:mon}

\noindent
As mentioned above, we regard continuous systems as arrows of type
\begin{flalign*} I \longrightarrow \coprod \limits_{d \in [0,\infty]}
O^{T_d}.
\end{flalign*} 

\noindent
In order to define them in $\topo$, we need to equip the target object
with a suitable topology.  A first choice would be the coproduct
topology (as suggested by the expression above), but this is not
suitable, since in many cases such a topology forbids the
system to change the duration of its evolutions along different
inputs. 

Let us thus explore an alternative topology; the strategy will be
similar to the one used in the definition of a \emph{Moore path
  category} where, given a topological space $X$, arrows are paths
(\ie, evolutions) $[0,d] \rightarrow X$ and composition corresponds to
the concatenation of those paths (\emph{cf.} \cite{ronnie2009}).
Actually, the flattening operation of monad $\MH$, discussed below,
can be seen as \emph{a more general version} of path concatenation.

Consider, with no loss of generality, that all evolutions have
domain $\Rz$. Such is possible when one notices that $T_d$ (for some
$d \in [0,\infty]$) is a \emph{retract} of $\Rz$ through the
\emph{truncation} function (the retraction)
\begin{flalign*} 
  \curlywedge_d : \Rz \longrightarrow T_d
\end{flalign*} 
\noindent
$\curlywedge_d \> \eqdef \hmul{id}{(\leq_d)}{\const{d}}$ and considers
just those functions $f \in O^{\Rz}$ that become constant after time
instant $d$, \ie $f \comp \curlywedge_d = f$. This gives a family of
bijections $\{ f \in O^{\Rz} \> | \> f \comp \curlywedge_d = f \}
\cong O^{T_d}$ indexed by durations $d \in [0,\infty]$.  Continuous
systems thus become arrows typed as,
\begin{flalign*} 
  I \longrightarrow \{ \> (f,d) \in O^{\Rz} \times [0,\infty] \> | \>
  f \comp \curlywedge_d = f \}
\end{flalign*}
where the target object comes equipped with the canonical topology.
This leads to the following definition for the underlying
functor of monad $\MH$.
\begin{Mdefinition}
$\MH : \topo \rightarrow \topo$ is a mapping such that
for any objects $X,Y \in |\topo|$ and any continuous
function $g : X \rightarrow Y$,
\begin{eqnarray*}
  & \MH X  & \eqdef \{ \> (f,d) \in X^{\emph{\Rz}} \times \emph{\Dur} \> | \>
        f \comp \curlywedge_d = f  \> \} \\
  & \MH g  & \eqdef g^{\emph{\Rz}} \times \> id  
\end{eqnarray*}
where $\emph{\Dur} = [0,\infty]$ is the one-point compactification of
$\mathbb{R}_{\geq 0}$ $($\emph{cf.} \emph{\cite{JGL-topology}}$)$, and
$g^{\emph{\Rz}} \> f = g \comp f$.
\end{Mdefinition}

\begin{Mtheorem}
  $\MH$ is a functor.
\end{Mtheorem}

\begin{proof}
  We need to show that for any continuous
  functions $g : X \rightarrow Y, h : Y \rightarrow Z$,
  $\MH g : \MH X \rightarrow \MH Y$ is continuous, and
  $\MH ( h \comp g ) = \MH h \comp \MH g$
  
  \noindent
  Since $\MH g = g^{\Rz} \times id$, and $g$ is continuous, then $\MH g$
  must be as well. Distributivity of composition follows from
  property
  \begin{flalign*}
    \iota \comp \MH g =  \big ( ( \> \_ \> \times \Dur) \comp (
     \> \_ \> )^{\Rz} \> g  \big ) \comp \iota
  \end{flalign*}
  \noindent
  where $\iota$ is the inclusion map $\MH X \hookrightarrow (X^{\Rz}
  \times \Dur)$, $(\> \_ \> \times \Dur)$ is the ($\Dur$) product
  functor, and $( \> \_ \> )^{\Rz}$ the ($\Rz$) exponential functor.
\end{proof}
\noindent
Let us explore some examples of continuous systems characterised 
as arrows $I \rightarrow \MH O$.

\begin{Mexample}
  \label{exsig}
  \emph{Signal generators} are classical examples of continuous
  systems that can generate sinusoidal waves as output. They can be
  regarded as arrows $s : \mathbb{R} \rightarrow \MH \mathbb{R}$ such
  that $s \> r \> \eqdef (r + (sin \> \_ \>), \> \infty)$.
\end{Mexample}
\noindent
Note that, in contrast to the coproduct topology (in the target
object), the topology chosen for $\MH$ allows durations to change, and
thus captures a wider range of behaviours. For example,
\begin{Mexample}
  \label{exther}
  Consider a \emph{thermostat} $c : \mathbb{R} \rightarrow \MH
  \mathbb{R}$ that, given the current temperature, \emph{linearly}
  raises it to, say, 20 {${}^\circ$}C. Such a behaviour can be
  expressed as $c \> r \> \eqdef ( (r + \> \_ \>) , \> 20 \circleddash
  r )$ where $\circleddash : \mathbb{R} \times \mathbb{R} \rightarrow
  \mathbb{R}$ is the truncated subtraction, \ie, $\circleddash =
  \hmul{(-)}{(>)}{\const{0}}$.
\end{Mexample}

\noindent
The execution time of system $c$ is thus inversely proportional to the
current temperature (which is given as input).

One may also consider another component that takes action after $c$,
and whose functionality is, for instance, to maintain the current
temperature.  The result is a composed system that can raise
temperatures to a desired level and then maintain them -- we will
explore this specific case in the next section. Of course, analogous
behaviour can also be found in \emph{e.g.}, cruise control systems,
water level regulators, and production lines. For example, imagine a
component of a cruise control system that gives control of the car's
velocity to another component whenever an obstacle is detected, or
the emergency mode becomes active. As we will see in the sequel, Kleisli
composition (for monad $\MH$) caters for this sort of action.

The following definition will help in the development of monad $\MH$.

\begin{Mdefinition}
  For any given topological space $X \in |\topo|$, define continuous
  function $\theta_X : \MH X \rightarrow X$ such that
  \begin{flalign*}
    \theta_X \> (f,d) \> \eqdef f \> 0.
  \end{flalign*}
\end{Mdefinition}

\noindent
Actually, we can canonically extend $\theta_X : \MH X \rightarrow X$
to a natural transformation $\theta : \MH \rightarrow Id$, since it is
straightforward to show that the following diagram commutes for any
continuous function $f : X \rightarrow Y$.
    \begin{center}
      \begin{tabular}{c}
        \xymatrix {
          \MH X
          \ar[r]^{ \MH f } \ar[d]_{ \theta_X } 
           & \MH Y \ar[d]^{ \theta_Y } \\
          X \ar[r]_{ f } & Y
        }
      \end{tabular}
    \end{center}

\noindent
Moreover, it becomes possible to express the first law of continuous
dynamical systems (recall the previous section) in a concise,
diagramatic manner: simply by saying that system $c : I \rightarrow
\MH I$ obeys the first law ((1) above) iff the diagram below commutes.
  \begin{center}
  \begin{tabular}{c}
    \xymatrix {
      I \ar[r]^c \ar[dr]_{id} & \MH I \ar[d]^{\theta_I} \\
      & I 
    }
  \end{tabular}
 \end{center}

\noindent
Actually, we can generalise the diagram to
  \begin{center}
  \begin{tabular}{c}
    \xymatrix {
      I' \ar[r]^c \ar@{^{(}->}[dr]_{\iota} & \MH I 
      \ar[d]^{\theta_I} \\
      & I 
    }
  \end{tabular}
 \end{center}
\noindent
where $\iota : I' \hookrightarrow I$ is the inclusion map $I'
\subseteq I$. We qualify as \emph{pre-dynamical} any system that
follows this generalised condition. Note that both examples above
(\ref{exsig} and \ref{exther}) concern pre-dynamical systems.

We shall now discuss how to equip $\MH$ with the structure of a monad.
As already mentioned, in programming semantics a monad captures a
behavioural effect and provides mechanisms to wrap a value into such
an effect and to flatten two effects into a single one. Technically,
they are referred to as the monad identity $\eta : Id \rightarrow
\MH$, and its multiplication $\mu : \MH \MH \rightarrow \MH$,
respectively.  Let us start by defining the unit operation $\eta : Id
\rightarrow \MH$, which will denote trivial evolutions.
\begin{Mdefinition}
\label{def:unit}
   Given a space $X \in |\topo|$, function $\eta_X : X
   \rightarrow \MH X$ is defined by 
   \begin{flalign*}
     \eta_X \> x \> \eqdef (\const{x}, \> 0).
   \end{flalign*}
\end{Mdefinition} 
\noindent
Intuitively, arrow $\eta_X : X \rightarrow \MH X$ defines a system
whose outputs are always trivial evolutions, \ie, with duration
zero. For this reason we will refer to $\eta_X$ as $copy_X$, and often
omit the subscript.

\begin{Mlemma}
  The mapping $\eta : Id \rightarrow \MH$ is a natural transformation,
  \ie,  for any topological space $X$, $\eta_X : X \rightarrow
  \MH X$ is a continuous function, and, moreover, the diagram below
  commutes
  \begin{center}
       \begin{tabular}{c }
               \xymatrix{
                X \ar[r]^h  \ar[d]_{\eta_X} & Y \ar[d]^{\eta_Y} \\
               \MH X \ar[r]_{ \MH h} & \MH Y 
           } 
       \end{tabular}
   \end{center}
   for any continuous function $h : X \rightarrow Y$.
\end{Mlemma}

\begin{proof}
  To see that $\eta_X$ is continuous, observe first that $\eta_X =
  \pv{\lambda \pi_1, \const{0}}$. Then,
  
  \begin{center}
     $\infer[ (\> \downarrow_r \>) ]{ \pv{ \lambda \pi_1,
         \underline{0} } : X \rightarrow \MH X } { \infer[ (\> \times \>) ]{
         \pv{ \lambda \pi_1, \underline{0} } : X \rightarrow X^{\Rz} \times \Dur }
       { \infer[ (\> \lambda \>) ]{ \lambda \pi_1 : X \rightarrow X^{\Rz} } {
           \pi_1 : X \times \Rz \rightarrow X } } }$
  \end{center}
  \noindent
  It remains to show the naturality of $\eta : Id \rightarrow \MH$.
  Consider the diagram
  \begin{center}
       \begin{tabular}{c }
               \xymatrix{
                x \ar@{|->}[rr]^h  \ar@{|->}[d]_{\eta_X} & & h \> x 
                \ar@{|->}[d]^{\eta_Y} \\
                (\underline{x}, 0) \ar@{|.>}[rr]_{h^{\Rz}
                  \times id}
                & & (\underline{h \> x}, 0) 
           } 
       \end{tabular}
   \end{center}
  \noindent
  where $h : X \rightarrow Y$ is an arbitrary continuous function.
  Property $h \comp \underline{x} = \underline{h \> x}$ entails its
  commutativity.
\end{proof}

\medskip

\noindent
It is also simple to see that, for any topological space $X \in |\topo|$, the
following diagram commutes
  \begin{center}
  \begin{tabular}{c}
    \xymatrix {
      X \ar[r]^{\eta_X} \ar[dr]_{id} & \MH X \ar[d]^{\theta_X} \\
      & X 
    }
  \end{tabular}
 \end{center}

\noindent
(\ie, that $\eta_X$ is pre-dynamical). 
Actually, this is one of two laws that characterise $\theta_X$ as an
\emph{Eilenberg-Moore} $\MH$-algebra \cite{cats}, a notion we will visit 
later in the paper.

The next step is to define multiplication $\mu : \MH \MH \rightarrow
\MH$. We start with an (auxiliary) definition of evolution (or path)
concatenation.
\begin{Mdefinition}
 Given any elements $(f,d), (g,e) \in \MH X$, define
 \begin{flalign*}
   (f,d) \conce \> (g,e) \> \eqdef (f \conce_d \> g, d + e)
 \end{flalign*}
 \noindent
 where $f \conce_d \> g \> \eqdef \hmul{f}{( \leq_d )}{g \> ( \> \_ 
   \> - d)}$.
\end{Mdefinition}
\noindent
Let us omit the subscript in $\conc_d$.  Note that $f \conc g$ is
continuous whenever the endpoint of $f$ and the startpoint of $g$
coincide. We will show that this condition is always met for the case
of multiplication.

\begin{Mdefinition}
\label{def:mult}
Given any topological space $X \in |\topo|$, define
\begin{flalign*}
  \mu_X \> (f,d) \> \eqdef \begin{cases}
    (\theta \comp f,d) \conce \> (f \> d) & \mbox{if } d \not = \infty \\
    (\theta \comp f, \> \infty) & \mbox{otherwise}
  \end{cases}
\end{flalign*}
\end{Mdefinition}

\noindent
Intuitively, multiplication will serve to concatenate the resulting
evolutions of two components.

\begin{Mlemma} 
  \label{lem:mult_nat}
  The family of mappings $\mu$ defines a natural tranformation.
\end{Mlemma}

\begin{proof}
  In appendix.
\end{proof}

\begin{Mlemma}
  For every topological space $X \in |\topo|$, the diagram below commutes
    \begin{center}
       \begin{tabular}{c }
               \xymatrix{
                \MH \MH X \ar[r]^{\mu_X}  \ar[d]_{\MH \theta_X} 
                & \MH X \ar[d]^{\theta_X} \\
               \MH X \ar[r]_{ \theta_X } & X
           } 
       \end{tabular}
   \end{center}
\end{Mlemma}

\begin{proof}
  Consider a pair $(f,d) \in \MH \MH X$, where $d$ is finite. Then,
  \begin{eqnarray*}
    && \theta \comp \mu \> (f,d) 
    \justl{=}{Definition of $\mu$}
    \theta \> ((\theta \comp f, d) \conc (f \> d) )
    \justl{=}{Definition of $\conc$ on point $0$}
    \theta  \> ( (\theta \comp f, d)  )
    \justl{=}{Definition of $\MH$ }
    \theta \comp \MH \theta \> (f,d)
  \end{eqnarray*}
  \noindent
  Proof for the case in which $d$ is infinite is achieved via an
  analogous reasoning process.
\end{proof}

\noindent
This property, together with the fact that $\theta_X \comp \eta_X =
id$ (discussed above), entail that $\theta_X : \MH X \rightarrow X$ is
an Eilenberg-Moore $\MH$-algebra.  In words, an algebra of functor
$\MH$ that is compatible with the monadic structure defined
above. This notion will be rather useful in the sequel.

\begin{Mtheorem}
  \label{theo:monad}
  $\pv{\MH,\eta,\mu}$ forms a monad.
\end{Mtheorem}

\begin{proof}
  In appendix.
\end{proof}

\section{ \dots and its Kleisli category ($\klH$) }
\label{sec:kl}

\noindent
If a monad abstracts a computational effect, its Kleisli category,
represents the universe of computations encapsulated in such an
effect. Hence, in the case of monad $\MH$, the associated Kleisli category of
$\MH$ $(\klH)$ provides an interesting setting to study the
requirements placed by continuity over different forms of
composition. Actually, the envisaged calculus of continuous, and
hybrid components is essentially its calculus. 

This section studies the Kleisli composition of $\klH$, and
illustrates its application to the specification of continuous systems
-- the hybrid ones will be discussed later in the paper.  We start
with the definition of $\klH$.

\begin{Mdefinition}
Category $\klH$ is defined as follows:

\begin{itemize}
 \item $| \klH| = | \topo |$,
 \item for any objects $I,O \in | \klH |$,
     ${\klH}(I,O) = {\topo}(I,\MH O)$,
 and for any object $I \in | \klH |$, $\eta_I$ is its identity.
 \item Given two
   arrows $c_1 : I \rightarrow \MH K$, 
   $c_2 : K \rightarrow \MH O$ their composition,
   denoted by $c_2 \kcomp c_1$, is given by
     $\mu \comp \MH c_2 \comp c_1$.
    Diagrammatically, 
   \begin{center}
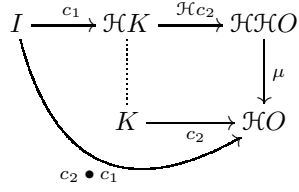
  \begin{tabular}{c}
  \xymatrix{
    I \ar[r]^{c_1} \ar@/_3pc/[drr]_{c_2 \kcomp c_1} 
    & \MH K \ar[r]^{\MH c_2} \ar@{..}[d] & \MH \MH O \ar[d]^{\mu} \\
    & K \ar[r]_{c_2} & \MH O
                 }
  \end{tabular}
  \end{center}

\end{itemize}

\end{Mdefinition}
\medskip
\noindent
Whenever found suitable, we will 
denote an arrow $c : I \rightarrow \MH O$ as $c : I \karrow
O$, and $\pi_1 \comp c$ as $f_c : I \rightarrow O^{\Rz}$.

Recall that arrows $c : I \karrow O$ are here interpreted as
continuous components, which means that the Kleisli composition of
$\MH$ can be seen as a component operator. Let us explore its
behaviour: consider two systems
\begin{flalign*}
c_1 : I \karrow K, \> \> \> c_2 : K \karrow O.
\end{flalign*}

\noindent
For a given input $x \in I$, compute the execution time of $c_2 \kcomp
c_1$,
\begin{eqnarray*}
  && \pi_2 \comp (c_2 \kcomp c_1) \> \> (x)
  \justl{=}{Kleisli composition}
  \pi_2 \comp \mu \comp \MH c_2 \comp c_1 \> \> (x) 
  \justl{=}{Definition of $\MH$, let $d = \pi_2 \comp c_1 \> (x)$ }
  \pi_2 \comp  \mu (c_2 \comp (f_{c_1} \> x), d)
  \justl{=}{Definition of $\mu$  }
  d + \pi_2 (c_2 \comp (f_{c_1} \> x) \> \> d)
  \justl{=}{ Composition }
  d + \pi_2 (c_2 \> (f_{c_1} \> x \> d))
\end{eqnarray*}

\noindent
This means that the execution time of $c_2 \kcomp c_1$ is the
  sum of the execution times of $c_1$ (for input $x$) and $c_2$
  (which receives value $f_{c_1} \> x \> d$ as input).  On the other
hand,
\begin{eqnarray*}
  &&  \pi_1 \comp (c_2 \kcomp c_1) \> \> (x)
  \justl{=}{Kleisli composition}
  \pi_1 \comp \mu \comp \MH c_2 \comp c_1 \> \> (x)
  \justl{=}{Definition of $\MH$, let $d = \pi_2 \comp c_1 \> (x)$ }
  \pi_1 \comp \mu \> (c_2 \comp (f_{c_1} \> x), d)
  \justl{=}{Definition of $\mu$}
  \theta \comp c_2 \comp (f_{c_1} \> x) \conc
  (f_{c_2} \> (f_{c_1} \> x \> d))
  \justl{=}{Definition of $\conc$}
  \hmul{ \theta \comp c_2 \> (f_{c_1} \> x \> \> \_ \> ) }
       { (\leq_d) }
       { f_{c_2} \> (f_{c_1} \> x \> d) \> ( \> \_ \>  - d) }
\end{eqnarray*}

\medskip
\noindent
Hence, if $c_2$ is pre-dynamical, 
\begin{flalign*}
  f_{(c_2 \kcomp c_1)} \> x =   
  \hmul{ (f_{c_1} \> x \> \> \_ \> ) }
       { (\leq_d) }
       { f_{c_2} \> (f_{c_1} \> x \> d) \> ( \> \_ \>  - d) }
\end{flalign*}

\noindent
The last expression tells that for the duration of $c_1 \> x$, $c_2
\kcomp c_1 \> x$ evolves first according to $c_1$, and then, on its
termination, according to $c_2$ which receives as input the endpoint
of $f_{c_1} \> x$.  Clearly, this is the expected behaviour according
to the definition of operation $\mu$, which `concatenates' evolutions.
Intuitively, $c_2 \kcomp c_1$ may also be described as mentioned in
Section 1: component $c_1$ acts and then, at instant $d$, gives
control of its evolution to $c_2$.

If, however, $c_2$ is not pre-dynamical, then up to completion of
interval $[0,d]$, $c_2$ \emph{`alters'} the evolution of $c_1$; then it
proceeds according to its own evolution.
These notions are illustrated in the following examples.

\begin{Mexample}
  \label{exFM}
  Given two signal generators 
  $c_1,c_2 : \mathbb{R} \karrow  \mathbb{R}$ defined
  as
  \begin{flalign*}
    & c_1 \> r \> \eqdef (r + (sin \> \> \_ \>), 3 \pi ), \> \> \>
    c_2 \> r \> \eqdef (r + sin \> (3 \times \> \_ \>), 3 \pi)
  \end{flalign*}
  the evolution $c_1 \kcomp (c_2 \kcomp c_1) \> 0$ is represented by
  the plot below.
\end{Mexample}

  \pgfplotsset{samples=100}

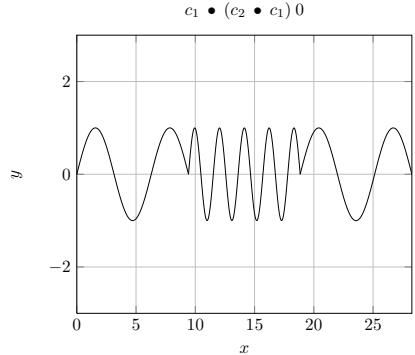
  \begin{center}
  \scalebox{0.65}{
    \begin{tikzpicture} 
    \begin{axis}[
        title={$c_1 \kcomp (c_2 \kcomp c_1) \> 0$},
        xlabel=$x$, ylabel=$y$,
        ymin= -3,
        ymax = 3,
        xmin = 0,
        xmax = 28.2735,
        grid = major
    ]
        \addplot[smooth,
        domain = 0:3*3.1415] { sin(deg(x)) };
        \addplot[smooth,shift = {(94.245,0.0)},
        domain = 0:3*3.1415] { sin(deg(3*3.1415)) + sin(3*(deg(x))) };
        \addplot[smooth,shift = {(188.49,0.0)},
        domain = 0:3*3.1415] { sin(deg(6*3.1415)) + sin((deg(x))) };
    \end{axis}
    \end{tikzpicture}
    }
  \end{center}

\noindent
This type of signal is commonly seen in \emph{frequency modulation}:
the varying frequency is used to encode information for
electromagnetic transmission. Note that $c_1$ gives control for some
time to $c_2$, and then `takes it back'.

In order to amplify signals, one can use component $a : \mathbb{R}
\karrow \mathbb{R}$, where $a \> r \> \eqdef (\const{r \times 2},
0)$ (note that since system $a$ is not pre-dynamical it can alter
evolutions of other components). Given input $0$,
system $c_1 \kcomp (c_2 \kcomp (a \kcomp c_1))$, returns the following
evolution.

  \begin{center}
  \scalebox{0.65}
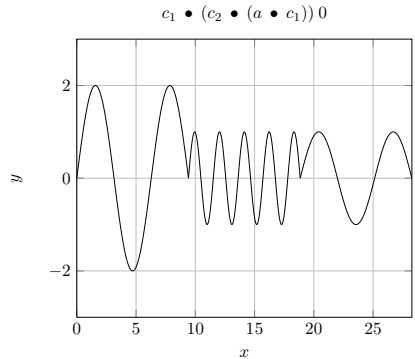
{
    \begin{tikzpicture} 
    \begin{axis}[
        title={$c_1 \kcomp (c_2 \kcomp (a \kcomp c_1)) \> 0$},
        xlabel=$x$, ylabel=$y$,
        ymin= -3,
        ymax = 3,
        xmin = 0,
        xmax = 28.2735,
        grid = major
    ]
        \addplot[smooth,
        domain = 0:3*3.1415] { 2 * sin(deg(x)) };
        \addplot[smooth,shift = {(94.245,0.0)},
        domain = 0:3*3.1415] { sin(deg(3*3.1415)) + sin(3*(deg(x))) };
        \addplot[smooth,shift = {(188.49,0.0)},
        domain = 0:3*3.1415] { sin(deg(6*3.1415)) + sin((deg(x))) };
    \end{axis}
    \end{tikzpicture}
    }
  \end{center}

\begin{Mexample}
Suppose the temperature of a room is to be regulated according to the
following discipline: starting at $10 \> {}^\circ$C, seek
to reach and maintain $20 \> {}^\circ$C, but in no case surpass $20.5
\> {}^\circ$C.  To realise such a system, three elementary components
have to work together: $c_1$ to raise the temperature to $20 \>
{}^\circ$C, component $c_2$ to maintain a given temperature, and
component $c_3$ to ensure the temperature never goes over $20.5 \>
{}^\circ$C. Formally,
\begin{flalign*}
  & c_1 \> x = ( \> (x + \> \_ \>), 
  \> 20 \circleddash x \>  ) \\ 
  & c_2 \> x =  ( \> x + (\sin \> \_ \>), 
  \> \infty \>  ) \\
  & c_3 \> x = 
  ( \> \hmul{\underline{x}}{(x \leq 20.5)}{\underline{20.5}} 
  \>, 0 \> )
\end{flalign*}
\end{Mexample}

\noindent
In a first try one may compose $c_2,c_1$ into $c_2 \kcomp c_1$. This
results in a component able to read the current temperature, raise it
to $20 \> {}^\circ$C, and then keep it stable, as exemplified by the
plot below.

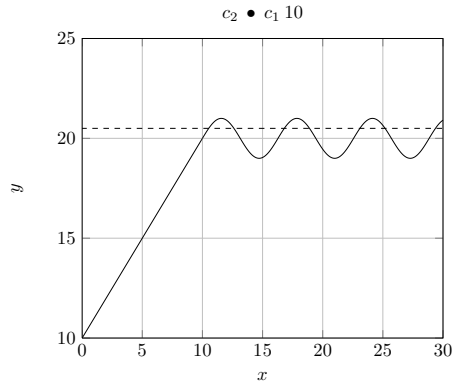
\begin{center}
\scalebox{0.7}{
  \begin{tikzpicture} 
  \begin{axis}[
      title={$c_2 \kcomp c_1 \> 10$},
      xlabel=$x$, ylabel=$y$,
      ymin= 10,
      ymax = 25,
      xmin = 0,
      xmax = 30,
      grid = major
  ]
      \addplot[smooth,
      domain = 0:10] { 10 + x };
      \addplot[smooth,
      domain = 10:30] { 20 + sin(deg(x - 10)) };
      \addplot[smooth,dashed,
      domain = -1 : 31] {20.5};
  \end{axis}
  \end{tikzpicture}
  }
\end{center}

\noindent
If, however, temperatures over $20.5\> {}^\circ$C occur, composition
$c_3 \kcomp (c_2 \kcomp c_1)$ puts the system back into the right
track as illustrated in the following plot.
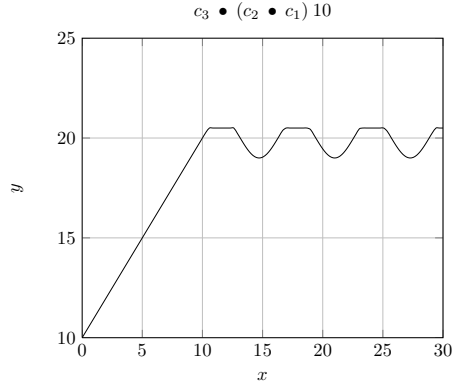
\pgfplotsset{samples=70}

\pgfmathdeclarefunction{f}{1}{%
  \pgfmathparse{ (20 + sin(deg(#1 - 10)) > 20.5)*20.5 +
    (20 + sin(deg(#1 - 10)) <= 20.5)*(20 + sin(deg(#1 - 10))) }%
}

\begin{center}
  \scalebox{0.7}{
  \begin{tikzpicture} 
  \begin{axis}[
      title={$c_3 \kcomp (c_2 \kcomp c_1) \> 10$},
      xlabel=$x$, ylabel=$y$,
      ymin= 10,
      ymax = 25,
      xmin = 0,
      xmax = 30,
      grid = major
  ]
      \addplot[smooth,
      domain = 0:10] { 10 + x };
      \addplot[smooth,
      domain = 10:30] {  ((20 + sin(deg(x - 10)) > 20.5)*20.5)
    + ((20 + sin(deg(x - 10)) < 20.5)*(20 + sin(deg(x - 10))))  };
  \end{axis}
  \end{tikzpicture}
}
\end{center}

\noindent
Clearly, $c_3$ can be regarded as a supervisor system that, for the
sake of efficiency, only acts when temperatures exceed the threshold,
using just enough power to keep the temperate below the limit.
Actually, note that $c_3$ is able to play a supervisory role
  precisely because it is non pre-dynamical.  Of course in this
specific case, we assume that $c_3$ has an idealised behaviour, which,
despite pedagogical, is quite unrealistic.

\noindent
The examples above hint at an interesting property of 
evolutions with infinite duration.

\begin{Mtheorem}
   Consider two arrows $c_1 : I \karrow O$,
  $c_2 : O \karrow  O$. If system $c_2$ is pre-dynamical
  and $\img \> (\pi_2 \comp c_1 \comp \iota) \subseteq
  \{ \infty \}$ for some embedding $\iota : I' \hookrightarrow I$, then
  \begin{flalign*}
    (c_2 \kcomp c_1) \comp \iota = c_1 \comp \iota 
  \end{flalign*}
\end{Mtheorem}

\begin{proof}

  \begin{eqnarray*}
    && (c_2 \kcomp c_1) \comp \iota
    \justl{=}{Kleisli composition, $\img \>
      (\pi_2 \comp c_1 \comp \iota) \subseteq
      \{ \infty \}$ }
    \mu \> (c_2 \comp (f_{c_1} \comp \iota ), \infty )
    \justl{=}{ Definition of $\mu$ }
    ( \theta \comp c_2 \comp (f_{c_1} \comp \iota ), \infty )
    \justl{=}{System $c_2$ is pre-dynamical }
    ( f_{c_1} \comp \iota , \infty )
    \justl{=}{ Notation }
    c_1 \comp \iota
  \end{eqnarray*}
\end{proof}

\begin{Mcorollary}
  If $c_2$ is pre-dynamical and $\img \> (\pi_2 \comp c_1) \subseteq
  \{ \infty \}$, then $ c_2 \kcomp c_1 = c_1$.
\end{Mcorollary}
\noindent
This means that if evolutions of the first component always exhibit an
infinite duration, the second one, if pre-dynamical, will never have
the chance to execute.

\medskip
\noindent
In general, $\MH$-Kleisli composition provides the basic composition
mechanism for continuous components; the structure of $\klH$ yields
its basic laws. To be more concrete, take $copy$ as the
trivial system that outputs its input with duration zero (\ie,
  the unit of monad $\MH$). Then, given systems $c_1,c_2,c_3$
\begin{flalign}
  & copy \> \kcomp \> c_1 = c_1 \\
  & c_1 \> \kcomp \> copy = c_1 \\
  & (c_3 \> \kcomp \> c_2) \> \kcomp 
  \> c_1 = c_3 \> \kcomp \> (c_2 \> \kcomp \>  c_1)
\end{flalign}

\section{Wiring mechanisms and (additional) composition operators}
\label{sec:Adj}
\noindent
In a category, (co)limits are a main tool to `build new arrows from
old ones', which in the case of $\klH$ translates to new forms of
component composition.  Actually, coproducts are easy to obtain
through the canonical adjunction between $\topo$ and $\klH$,

  \begin{center}
  \begin{tabular}{c}
  \xymatrix{
    \topo \ar@/^1pc/[rr]^{L} & {\> \> \perp}  & \klH
    \ar@/^1pc/[ll]^{R}
                 }
  \end{tabular}
  \end{center}

  \noindent
  which entails that $\klH$ inherits colimits of $\topo$ through
  $L$. For notational simplicity, given a continuous function $f : X
  \rightarrow Y$, we will denote system $L f = \eta \comp f 
  : X \karrow Y$ by $\corn{ f }{} $.

  \medskip
  \noindent
  In $\klH$, the coproduct (also known as a \emph{choice} operator) is
  inherited as follows: given two components

  \begin{center}
  \begin{tabular}{c}
    \xymatrix{ 
             I_1 \ar@{->}[dr]_{c_1}|-\xykarrow && 
             I_2 \ar@{->}[dl]^{c_2}|-\xykarrow
             \\
             & O &
           } 
    \end{tabular}
\end{center}

\noindent
define component $[c_1,c_2] : I_1 + I_2 \karrow O$ 
which makes the following diagram to commute.
  \begin{center}
  \begin{tabular}{c}
    \xymatrix{ 
      I_1 \ar@{->}[rr]^{ \corn{ i_1 } }|-\xykarrow
      \ar@{->}[drr]_{c_1}|-\xykarrow
      && I_1 + I_2 \ar@{->}[d]^{[c_1,c_2]}|-\xykarrow &&
      I_2 \ar@{->}[ll]_{ \corn{ i_2 } }|-\xykarrow 
      \ar@{->}[dll]^{c_2}|-\xykarrow \\
      && O &&
    } 
    \end{tabular}
\end{center}

\noindent
Intuitively, $[c_1,c_2]$ behaves as $c_1$ whenever input $I_1$ is
chosen, and as $c_2$ otherwise.  Such a mechanism is useful to
aggregate systems with the same codomain; the result being a
\emph{singular} system with different modes of operation
(corresponding to the respective subcomponents), chosen according to
the input received.  As usual, a functorial \emph{sum} operator is
easily defined.

\begin{Mdefinition}
  Consider components $c_1 : I_1 \karrow O_1$, $c_2 : I_2
  \karrow O_2$. Then define component $
  c_1 \boxplus c_2 : I_1 + I_2 \karrow O_1 + O_2$ as
  \begin{flalign*}
   & c_1 \boxplus c_2 \> \eqdef [ \corn{ i_1  } \kcomp c_1, 
  \corn{ i_2} \kcomp c_2 ]
  \end{flalign*}
\end{Mdefinition}

\noindent
The definition of operator choice  as the coproduct
universal arrow in $\klH$, yields a number of useful laws for free.
  \begin{flalign}
    & c_3 \kcomp [c_1, c_2] =  [c_3 \kcomp c_1, c_3 \kcomp c_2] \\
    & (c_1 \boxplus c_2) \kcomp \corn{i_1} = \corn{i_1} \kcomp c_1 \\
    & (c_1 \boxplus c_2) \kcomp \corn{i_2} = \corn{i_2} \kcomp c_2 \\
    & copy_X \boxplus copy_Y = copy_{X + Y} \\
    & (d_1 \boxplus d_2) \kcomp (c_1 \boxplus c_2) = 
    (d_1 \kcomp c_1) \boxplus (d_2 \kcomp c_2) \\
    & [d_1, d_2] \kcomp (c_1 \boxplus c_2) = 
    [d_1 \kcomp c_1, d_2 \kcomp c_2]
  \end{flalign}

\noindent
Moreover,
\begin{Mlemma}
  \label{inter_sum}
  For any continuous functions $f: X_1 \rightarrow Y_1, g : X_2
\rightarrow Y_2$, the following equation holds
  \begin{flalign}
    \corn{f} \boxplus \corn{g} = 
    \corn{f + g }
  \end{flalign}
\end{Mlemma}

\begin{proof}
  \begin{eqnarray*}
    && \corn{f} \boxplus \corn{g} 
    \justl{=}{Definition of $\boxplus$}
    [\corn{i_1} \kcomp \corn{f} , \corn{i_2} \kcomp \corn{g} ]
    \justl{=}{$L$ is a functor}
    [\corn{i_1 \comp f} , \corn{i_2 \comp g} ]
    \justl{=}{Definition of $L$}
    [copy \comp i_1 \comp f , copy \comp i_2 \comp g ]
    \justl{=}{ Universal property of coproduct }
    copy \comp [ i_1 \comp f , i_2 \comp g ]
    \justl{=}{Definition of $+$, definition of $L$}
    \corn{ f + g }
  \end{eqnarray*}
\end{proof}

\medskip
\noindent
The left adjoint is also useful to lift functions to the universe of
$\klH$.  This provides a number of interesting operations and wiring
mechanisms.  For example, recall the diagonal function $\vartriangle :
X \rightarrow X \times X $ which duplicates the input value; the
corresponding lifted version $\corn{\vartriangle} : X \karrow X \times
X$ duplicates evolutions.  Take now the scalar multiplication $*_s :
\mathbb{R} \rightarrow \mathbb{R}$; operation $\corn { *_s } :
\mathbb{R} \karrow \mathbb{R}$ can be used to amplify signals, a
ubiquitous procedure both in signal and control theory. Another
example is $\corn{ \pi_1 } : X \times Y \karrow X$ (resp. $\corn{
  \pi_2 } : X \times Y \karrow X$ ) which eliminates the right
(resp. left) side of `paired' evolutions.  Finally, $\corn{sw} : X
\times Y \karrow Y \times X$ swaps the order of evolutions, a
functionality graphically represented by wire swapping.

Since $L$ is a functor, the following
laws also come for free
  \begin{flalign}
    & \corn{id} = copy  \label{idfunc_F} \\
    &   \corn{ g } \kcomp \corn{ f }  = 
    \corn { g \comp f } \label{func_F}
  \end{flalign}

\medskip
\noindent
Finding limits in a Kleisli category through left adjoint $L$ is often
more difficult.  However, under specific conditions, $L$ also
preserves limits. The following theorem makes such conditions precise.

\begin{Mtheorem}
  \label{theo_kadj}
  Consider the Kleisli adjunction $L \dashv R$ of a given monad
  $\pv{\mathcal{T},\eta,\mu}$. Functor $L$ preserves whatever
  limits $\mathcal{T}$ does.
\end{Mtheorem}

\begin{proof}
  Observe the diagram
    \begin{center}
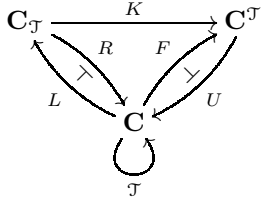
    \begin{tabular}{c}
      \xy 
      (7.6,-7.0)*={\rotatebox[origin=c]{45}{$\dashv$}},
      (21.8,-6.6)*={\rotatebox[origin=c]{-45}{$\dashv$}};
      \xymatrix {
         \mathbf{C}_{\mathcal{T}} \ar[rr]^{K} 
         \ar@/^0.5pc/[dr]^{R}
         && \mathbf{C}^{\mathcal{T}} \ar@/^0.5pc/[dl]^{U} \\
         & \mathbf{C} 
           \ar@/^0.5pc/[ul]^{L}
           \ar@/^0.5pc/[ur]^{F}
           \ar@(dl,dr)_{\mathcal{T}}
         & 
      }
      \endxy
    \end{tabular}
    \end{center}
    \noindent
    where $\mathbf{C}^{\mathcal{T}}$ is the Eilenberg-Moore category
    for monad $\mathcal{T}$ \cite{cats}, and $K$ the corresponding
    (\emph{fully faithful}) functor such that $\mathcal{T} = U K L$.
    Then, consider a limit $\lim_{\leftarrow} D$ in $\mathbf{C}$ and
    assume that $\mathcal{T}$ preserves it. This means that
    $\mathcal{T}( \lim_{\leftarrow} D )$ is the limit of $\mathcal{T}
    D$, and equivalently, $UKL( \lim_{\leftarrow} D )$ is the limit of
    $UKL D$. Since both $U$ and $K$ reflect limits, $L(
    \lim_{\leftarrow} D )$ must be the limit of $ L D$.
\end{proof}

\noindent
Note that the theorem above was stated in general terms and is thus
applicable to any monad. Even though easily proved, its consequences
are quite useful. For example, in the case of $\MH$ it provides 
pullbacks in $\klH$, as

\begin{Mtheorem}
  \label{theo_pull}
  Functor $\MH$ preserves pullbacks.
\end{Mtheorem}

\begin{proof}
  In the appendix.
\end{proof}

\noindent
More concretely, theorems \ref{theo_kadj} and \ref{theo_pull} assert
that any cospan $A \stackrel{f}{\rightarrow} C
\stackrel{g}{\leftarrow} B$ in $\topo$ gives rise to a pullback in
$\klH$, diagrammatically described as
    \begin{center}
      \begin{tabular}{c}
        \xymatrix {
          A \times_C B 
          \ar[r]^{\corn{ \pi_2 } }|-\xykarrow 
          \ar[d]_{ \corn { \pi_1 } }|-\xykarrow 
          \pullbackcorner & B \ar[d]^{ \corn { g } }|-\xykarrow \\
          A \ar[r]_{ \corn { f } }|-\xykarrow & C
        }
      \end{tabular}
    \end{center}

\noindent
One interesting cospan, worthy of special attention, is $A
\stackrel{!}{\rightarrow} 1 \stackrel{!}{\leftarrow} B$, which induces
the pullback
    \begin{center}
      \begin{tabular}{c}
        \xymatrix {
          A \times_1 B 
          \ar[r]^{\corn{ \pi_2 } }|\xykarrow
          \ar[d]_{ \corn { \pi_1 } }|\xykarrow 
          \pullbackcorner & B \ar[d]^{ \corn { ! } }|\xykarrow \\
          A \ar[r]_{ \corn { ! } }|\xykarrow & 1
        }
      \end{tabular}
    \end{center}

\noindent
Indeed, such a construction brings parallelism up front, and moreover,
makes possible to combine evolutions. More
concretely, the diagram states that whenever two systems are
\emph{compatible} -- in the sense that for any input they produce
evolutions with equal duration -- a new component that encapsulates
their parallel composition can be defined. Formally, two systems
$c_1 : I \karrow A$, $c_2 : I \karrow B$ are called compatible when
the diagram
    \begin{center}
      \begin{tabular}{c}
        \xymatrix {
          I
          \ar[r]^{ c_2 }|\xykarrow \ar[d]_{ c_1 }|\xykarrow
           & B \ar[d]^{ \corn { ! } }|\xykarrow \\
          A \ar[r]_{ \corn { ! } }|\xykarrow & 1
        }
      \end{tabular}
    \end{center}
\noindent
commutes (note that this is not trivially true, because $1$ is not a
final object in $\klH$). Then, let $E$ denote set $\{ \> ((f,d),(g,e))
\in \MH A \times \MH B \> | \> d = e \> \}$.  When the two systems are
compatible, a new component $\pulb{c_1,c_2} : I \karrow (A \times_1
B)$ comes forward through the mediating arrow (of the pullback), as
follows
\begin{flalign*}
  \pv{\pv{c_1,c_2}} \> \eqdef \gamma \comp \pv{c_1,c_2}
\end{flalign*}
\noindent
where $I \typ{\pv{c_1,c_2} {}} E \typ{\gamma}{\MH (A \times_1 B)}$, 
$ \> \> \gamma \> ((f,d),(g,d)) \> \eqdef ( \pv{f,g},d)$.

\medskip
\noindent
Note that $\img \> \pv{c_1,c_2} \subseteq E$ precisely because of the
assumption of compatibility between components (\emph{cf.} proof of
Theorem \ref{theo_pull}). In order to keep notation simple, we will
omit the $1$ in the subscript of $(A \times_1 B)$.

We call $\pulb{c_1,c_2}$ the \emph{strict parallel} composition of $c_1$ and
$c_2$. Let us illustrate its behaviour through a number of examples.
\begin{Mexample}
  \label{exam_com}
  Consider two signal generators, 
  \begin{flalign*}
  c_1 \> x \> = \>  ( \> x + ( \sin \> \_ \> ), \> 20 \> ), \> \> 
  c_2 \> x \> = \>  ( \> x + \sin \> (3 \times \_ \> ) , \> 20 \> )
  \end{flalign*}
  For input $0$, system $\pv{\pv{c_1,c_2}}$ exhibits the following
  behaviour

  \pgfplotsset{samples=100}

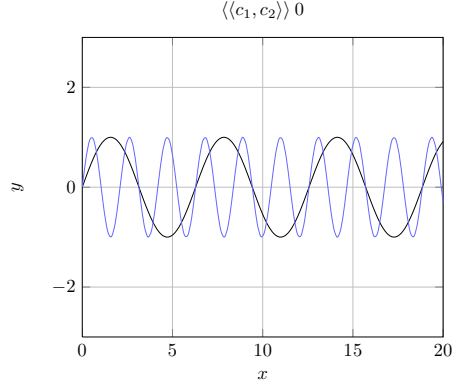
  \begin{center}
    \scalebox{0.7}{
    \begin{tikzpicture} 
    \begin{axis}[
        title={$\pv{\pv{c_1,c_2}} \> 0$},
        xlabel=$x$, ylabel=$y$,
        ymin= -3,
        ymax = 3,
        xmin = 0,
        xmax = 20,
        grid = major
    ]
        \addplot[smooth,
        domain = 0:20] { sin(deg(x)) };
        \addplot[smooth, blue!60!white,
        domain = 0:20] { sin(3*(deg(x))) };
    \end{axis}
    \end{tikzpicture}
    }
  \end{center}
  
  \noindent
  Consider now component $\corn{+} : \mathbb{R} \times \mathbb{R}
  \rightarrow \mathbb{R}$ which adds incoming signals. Then, for input
  $0$, the composed system $\corn{+} \kcomp \pv{\pv{c_1,c_2}}$ yields
  the following signal.
  \pgfplotsset{samples=100}

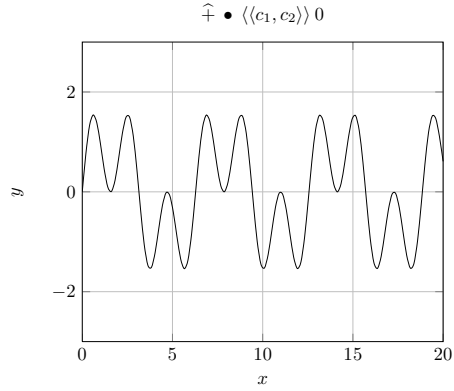
  \begin{center}
    \scalebox{0.7}{
    \begin{tikzpicture} 
    \begin{axis}[
        title={$\corn{+} \kcomp \pv{\pv{c_1,c_2}} \> 0$},
        xlabel=$x$, ylabel=$y$,
        ymin= -3,
        ymax = 3,
        xmin = 0,
        xmax = 20,
        grid = major
    ]
        \addplot[smooth,
        domain = 0:20] { sin(deg(x)) +  sin(3*(deg(x)))};
    \end{axis}
    \end{tikzpicture}
    }
  \end{center}
\end{Mexample}

\medskip
\noindent
Since strict parallelism comes from a pullback, the following operator
arises in a canonical way.
\begin{Mdefinition}
  Consider two
  continuous systems $c_1 : I_1 \karrow O_1$, $c_2 : I_2
  \karrow O_2$ such that $c_1 \kcomp \widehat{\pi_1}$ and 
  $c_2 \kcomp \widehat{\pi_2}$ are compatible. Then, define $
  c_1 \boxtimes c_2 : I_1 \times I_2 \karrow O_1 \times O_2$ as
  \begin{flalign*}
   & c_1 \boxtimes c_2 \> \eqdef \pulb{c_1 \kcomp \corn{\pi_1} , 
      c_2 \kcomp \corn{\pi_2}}
  \end{flalign*}
\end{Mdefinition}

\medskip
\noindent
Moreover, the following laws come for free, further contributing to an
emerging calculus of continuous and hybrid components: in each
  equation below, assume that both its sides are well defined (\ie\
  that the compatibility conditions are respected). Then, we have,
\begin{flalign}
 & \pulb{c_1,c_2} \kcomp d = \pulb{c_1 \kcomp d, c_2 \kcomp d} \\
 & \corn{ \pi_1 } \kcomp (c_1 \boxtimes c_2) = c_1 \kcomp \corn{ \pi_1 } \\
 & \corn{ \pi_2 } \kcomp (c_1 \boxtimes c_2) = c_2 \kcomp \corn{ \pi_2 } \\
 & \pulb{c_1,c_2} = (c_1 \boxtimes c_2) \kcomp \corn{\vartriangle} \\
 & copy_X \boxtimes copy_Y = copy_{X \times Y} \\
 & (d_1 \boxtimes d_2) \kcomp (c_1 \boxtimes c_2) = 
 (d_1 \kcomp c_1) \boxtimes (d_2 \kcomp c_2) \\
 & (d_1 \boxtimes d_2) \kcomp \pulb{c_1,c_2} = 
 \pulb{d_1 \kcomp c_1, d_2 \kcomp c_2}
\end{flalign}

\medskip

\noindent
Strict parallelism yields a result dual to Lemma
\ref{inter_sum}.

\begin{Mlemma}
  \label{exch_times}
  For any continuous functions $f: X_1 \rightarrow Y_1, g : X_2
  \rightarrow Y_2$, the following equation holds
  \begin{flalign}
    \corn{f}  \boxtimes \corn{g} = \corn{ f \times g }
  \end{flalign}
\end{Mlemma}

\begin{proof}
    \begin{eqnarray*}
    && \corn{f} \boxtimes \corn{g} 
    \justl{=}{Definition of $\boxtimes$}
    \pulb{ \corn{f} \kcomp \corn{ \pi_1 } , \corn{g} \kcomp \corn{\pi_2} }
    \justl{=}{ $L$ is a functor }
    \pulb{ \corn{ f \comp \pi_1 } , \corn{ g \comp \pi_2} }
    \justl{=}{Definition of $L$, $\times_1$ (in $\klH$) }
    \gamma \comp \langle
    \eta \comp f \comp \pi_1 , \eta \comp g \comp \pi_2 \rangle
    \justl{=}{Universal property of product (in $\mathbf{Top}$)}
    \gamma \comp (\eta \times \eta) \comp 
    \langle f \comp \pi_1 , g \comp \pi_2 \rangle 
    \justl{=}{ $\gamma \comp (\eta_{Y_1} \times \eta_{Y_2}) 
      = \eta_{Y_1 \times Y_2}$,
    definition of $\times$ (in $\topo$) }
    \eta \comp (f \times g) 
    \justl{=}{Definition of $L$}
    \corn{ f \times g }
  \end{eqnarray*}
\end{proof}

\medskip
\noindent
In some cases, however, putting two components in strict parallel may
be too restrictive or not enough to meet the system's design
requirements. The next section introduces a more relaxed version of
parallelism where synchronisation comes into play. Mathematically, our
construction explores the monoidal nature of functor $\MH$.

\section{Synchronised product and feedback}
\label{sec:Sync}

\noindent
Synchronised parallelism is a form of composition in which components
no longer need to be compatible in order to be put in
parallel. Instead, each of them can change the duration of the
corresponding evolutions according to the behaviour of the other. The
price to be paid is that the previous pullback (or any limit in
general) is no longer a suitable formalisation. Actually, adding a
monoidal structure \cite{seal2013} to functor $\MH$, as we will see in
the sequel, seems to be a better alternative.

\begin{Mdefinition}
  We say that functor $\MH$ is monoidal $($with respect to $\times$$)$
  if it comes equipped with a morphism $m : 1 \rightarrow \MH 1$, and
  a natural transformation $\delta : \MH \times \MH \rightarrow \MH$
  that make the following diagrams to commute for any topological
  spaces $X,Y \in |\topo|$.

  \begin{center}
  \begin{tabular}{c}
    \xymatrix {
      (\MH X \times \MH Y) \times \MH Z
      \ar[rr]^{\alpha}
      \ar[d]_{\delta \times id}
      && 
      \MH X \times ( \MH Y \times \MH Z)
      \ar[d]^{id \times \delta} \\
      \MH (X \times Y) \times \MH Z
      \ar[d]_{\delta} && 
      \MH X \times \MH (Y \times Z )
      \ar[d]^{\delta} \\
      \MH ((X \times Y) \times Z)
      \ar[rr]_{\MH \alpha }
      && \MH ( X \times  ( Y \times Z ))
    }
  \end{tabular}
  \end{center}

  \begin{center}
  \begin{tabular}{c c c c}
       \xymatrix{
           \MH X \times 1 \ar[rr]^{id \times m} 
           \ar[d]_{\pi_1}
           && \MH X \times \MH 1 \ar[d]^{\delta} \\
           \MH X && \MH (X \times 1) 
           \ar[ll]^{\MH \pi_1}
     } 
          & & & 
     \xymatrix{
           1 \times
           \MH X \ar[rr]^{m \times id}  \ar[d]_{\pi_2} 
           && \MH 1 \times \MH X \ar[d]^{\delta} \\
           \MH X && \ar[ll]^{\MH \pi_2} 
           \MH (1 \times X)
     } 
   \end{tabular}
   \end{center}

\end{Mdefinition}

\noindent
Hence, functor $\MH$ can be made monoidal once a suitable morphism
$m : 1 \rightarrow \mathcal{H} 1$ and a natural transformation
$\delta : \mathcal{H} \times \mathcal{H} \rightarrow \mathcal{H}$ are
defined.
\begin{Mdefinition}
  Let us define such mappings as
  \begin{flalign*}
    & m \> \eqdef copy \\
    & \delta_{X,Y} \> ((f,d),(g,e)) \> \eqdef (\pv{f,g}, d \curlyvee e)
  \end{flalign*}
  \noindent
  where continuous function
  $\curlyvee : \emph{\Dur} \times \emph{\Dur} \rightarrow \emph{\Dur}$
  is defined as $\curlyvee \> \eqdef \hmul{\pi_1}{(\geq)}{\pi_2}$.
\end{Mdefinition}

\noindent
As a side note, observe that a possible definition of $\delta$ resorts
to the minimum function $\curlywedge$ (instead of $\curlyvee$) but
then the diagrams above would not commute. Indeed, for such an
alternative to work, $m$ would need to be changed into a variant of
function $copy$ whose evolutions are always infinite.

\begin{Mlemma}
  \label{lemdelta}
  $\delta$ is a natural transformation. 

\end{Mlemma}

\begin{proof}
  We know that function $\delta : \MH X \times \MH Y \rightarrow \MH (X \times Y)$
  is defined as,
  \begin{flalign*}
    \MH X \times \MH Y \stackrel{}{\longrightarrow} X^{\Rz} 
    \times Y^{\Rz} \times \Dur \times \Dur
    \stackrel{\cong}{\longrightarrow} (X \times Y)^{\Rz} \times \Dur \times \Dur
    \stackrel{id \times \curlyvee}{\longrightarrow} \MH (X \times Y).
  \end{flalign*}
  \noindent
  Since $\curlyvee : \Dur \times \Dur \rightarrow \Dur$ is
  continuous, $\delta : \MH X \times \MH Y \rightarrow \MH (X \times Y)$
  must be continuous as well. To show that the naturality property holds, we reason
  \begin{eqnarray*}
    && \MH ( a \times b ) \comp \delta \> ((f,d),(g,e)) 
    \justl{=}{ Definition of $\MH$ and $\delta$}
    ( (a \times b) \comp \pv{f,g}, \> d \curlyvee e )
    \justl{=}{ Universal property of product }
    ( \pv{a \comp f, b \comp g}, \> d \curlyvee e )    
    \justl{=}{ Definition of $\delta$ }
    \delta \> ((a \comp f, d), (b \comp g, e))
    \justl{=}{ Definition of $\MH$ }
    \delta \comp (\MH a \times \MH b) \> ((f,d),(g,e))
  \end{eqnarray*}
\end{proof}

\noindent
We can now state the expected result.

\begin{Mtheorem}
  \label{teodelta}
  When equipped with natural transformation $\delta$ and morphism
  $m$, $\MH$ is a monoidal functor.


\end{Mtheorem}

\begin{proof}
  In appendix.
\end{proof}

\noindent
The monoidal structure $\pv{ \MH, \delta, m}$ defines a specific operator
for synchronised parallelism, which behaves as follows: given two
components with the same domain $c_1 : I \rightarrow \MH A, \> \> \>
c_2 : I \rightarrow \MH B$, define $\delta \comp \pv{c_1,c_2} : I
\rightarrow \MH A \times \MH B \rightarrow \MH (A \times B)$, to be
denoted in sequel by $\psync{c_1,c_2}$.  

System $\psync{c_1,c_2}$ runs $c_1$ and $c_2$ in parallel; however,
if one finishes earlier than the other, it is forced to stall its
evolution so that both components end at the same time.  In other
words, the duration of the shorter evolution is increased by keeping
it constant until the longer evolution
terminates.

Again, this form of parallelism is a lax version of strict
  parallelism, the cost being that many laws that hold before are now
lost.  Nevertheless, the monoidal structure of $\MH$ still makes
straightforward to show the following properties.
    \begin{flalign}
    & \corn{ f \times g } \kcomp \psync{c_1,c_2} = 
    \psync{\corn{f} \kcomp c_1, \corn{ g } 
      \kcomp c_2} \label{weakAbs} \\
    & \corn{ \alpha }
    \kcomp \llparenthesis \psync{c_1,c_2},c_3 \rrparenthesis  =
    \psync{c_1, \psync{c_2,c_3}}   \label{H_monoid} \\
    & \corn{ \pi_1 } \kcomp \psync{c, copy} = c   \label{H_unit} \\
    & \corn{ \pi_2 } \kcomp \psync{copy, c} = c   \label{H_unit2}
  \end{flalign}





\noindent
Moreover, we are able to canonically define a new operator, following
a path similar to the one used to define $\boxplus$ and $\boxtimes$.
\begin{Mdefinition}
  Given systems $c_1 : I_1 \karrow O_1$, $c_2 : I_2 \karrow O_2$,
  component $c_1 \> \boxed{s} \> c_2 : I_1 \times I_2 \karrow
  O_1 \times O_2$ is defined by
    \begin{flalign*}
    c_1 \> \boxed{s} \> c_2 \> \eqdef \psync{c_1 \kcomp \corn { \pi_1 } \>, 
      c_2 \kcomp \corn{ \pi_2 } \> }
  \end{flalign*}
\end{Mdefinition}
\medskip
\noindent
The following laws arise from routine calculations
\begin{flalign}
  & \corn{sw} \kcomp (c_2 \> \boxed{s} \> c_1) = 
  (c_1 \> \boxed{s} \> c_2) \comp sw \\
  & \corn{\alpha} \kcomp 
  ((c_1 \> \boxed{s} \> c_2) \> \boxed{s} \>  c_3) = 
  (c_1 \> \boxed{s} \> (c_2 \> \boxed{s} \>  c_3)) \comp \alpha \\
  & copy_X \> \boxed{s} \> copy_Y = copy_{X \times Y} \\
  & \corn{f}  \> \boxed{s} \> \corn{g} = 
  \corn{f \times g} 
\end{flalign}

\medskip
\noindent
Note that strict and synchronised parallel composition behave
identically but with one exception: in any given execution, the latter
increases the execution time of a system that finishes earlier than
the other. Hence, for compatible components both operators behave
exactly in the same way, and, therefore, the former inherits all laws
derived in this section for the latter.

Next, we introduce iteration for continuous systems. This facilitates
component specification and, moreover, can be used to express (or
detect) \emph{Zeno} behaviour \cite{alur2015}.
\begin{Mdefinition}
  Given a component $c : X \karrow X$, component $c^n : X
  \karrow X$ is defined by the $($Kleisli$)$ composition of
  $c$ with itself $n$ times. Formally,
  \begin{flalign*}
    & c^0 \> \eqdef copy, \> \>  \>
    c^n \> \eqdef c^{n - 1} \kcomp c
  \end{flalign*}
\end{Mdefinition}

\noindent
It is straightforward to check that the following equations
hold.
\begin{flalign}
  & copy^n = copy \\
  & c^1 = c \\
  & (c^n)^m = c^{n \times m} \\
  & c^n \kcomp c^m = c^{n + m} \\
  & (c \boxplus d)^n = c^n \boxplus d^n \\
  & (c \boxtimes d)^n = c^n \boxtimes d^n
\end{flalign}

\medskip
\noindent
Infinite iteration leads to the familiar notion of \emph{feedback}.
\begin{Mdefinition}
  Let $(X,d)$ be a complete metric space, and $c : X \karrow X$ a
  pre-dynamical system; denote the series $(\pi_2 \comp c^i (x))_{i
    \in \mathbb{N}}$ by $(s_i)_{i \in \mathbb{N}}$, and the sequence
  $(\pi_1 \comp c^i (x))_{i \in \mathbb{N}}$ by $(f_i)_{i \in
    \mathbb{N}}$.

  Then, assume that for any $x \in X$ whenever the series 
  $(s_i)_{i \in \mathbb{N}}$ converges the sequence 
  $(f_i)_{i \in \mathbb{N}}$ is Cauchy. More concretely, its
  elements get progressively closer to each other with respect
  to the metric,
  \begin{flalign*}
   d^*(g,h) \> \> \eqdef \sup_{t \in \emph{\Rz}} \> d(g(t), h(t)) .
  \end{flalign*}

  \noindent
  The interested reader will find in \cite{Kelley} more details about this
  metric.

  Finally, define infinite iteration $($$X \stackrel{c}{\karrow} X
  \stackrel{c}{\karrow} X \stackrel{c}{\karrow} \dots$$)$ as $\nu c :
  X \karrow X$ where
  \begin{flalign*}
    \pi_2 \comp \nu c \> (x) \> \eqdef 
    \begin{cases}
      \infty & \mbox{ if the series } (s_i)_{i \in \mathbb{N}} 
      \mbox{ diverges } \\
      \lim_{i \rightarrow \infty} s_i & \mbox{ otherwise }
    \end{cases}  
  \end{flalign*}
  \begin{flalign*}
    \hspace{-0.8cm}
    (\pi_1 \comp \nu c \> (x)) \> t \> \eqdef
    \begin{cases}
      f_k \> t & \mbox{ if } t < (\pi_2 \comp \nu c \> (x)) \\
      \big (\lim_{i \rightarrow \infty} f_i \big ) \> t & \mbox{ otherwise }
    \end{cases}
  \end{flalign*}
  \noindent
  for $k$ the smallest value such that $t \leq s_k$.
\end{Mdefinition}

\noindent
Intuitively, to compute the value at a certain instant $(t)$ in the
evolution $(\pi_1 \comp \nu c \> (x))$, we need to compose $c$ with
itself the necessary number of times for the composite `to reach that
instant'; only then it is possible to extract the value. To be
concrete, if each iteration of $c$ has two seconds of duration, to
calculate the value at five seconds in the evolution $(\pi_1 \comp \nu
c \> (x))$, we consider the composite $c^3$ and compute the expression
$(\pi_1 \comp c^3 \> (x)) \> 5$.

Observe that, since $c$ is pre-dynamical, the calculated value is not
changed by additional iterations, \ie
\begin{flalign*}
 (\pi_1 \comp c^k \> (x)) \> t = (\pi_1 \comp (c \kcomp c^k) \> (x)) \> t .
\end{flalign*}

\noindent
Actually, in the definition above one may forget the assumption of
$c$ being pre-dynamical as long as it is ensured that the sequence
$(f_i)_{i \in \mathbb{N}}$ is always Cauchy.

The following section gives concrete examples of parallel operators
and (infinite) iteration at work. The role of feedback in handling
Zeno behaviour is illustrated as well.

\section{From continuous to hybrid systems}
\label{sec:Str}

\noindent
Having characterised a calculus of continuous components based on the
structure of the Kleisli category of monad $\MH$, the next step is to
broaden the picture in order to handle systems that exhibit continuous
and discrete behaviour intertwined. Such is the purpose of this
section.  A number of examples will illustrate the approach proposed
here as well as some of the operators introduced in the previous
sections.  

Our aim is to  equip continuous systems with an
(internal) state space that behaves in a discrete manner. Therefore, arrows
become typed as
\begin{flalign*}
  & S \times I \longrightarrow S \times \MH O.
\end{flalign*}
\noindent
Intuitively, given a state ($s \in S$) and an input ($i \in
I$), the component transits (internally) into another state and
presents continuous evolutions that can be directly observed. This
gets us closer to the notion of hybrid system, as a family of
continuous systems indexed by a state space. On the other hand, this
approach is aligned with the notion of components as coalgebras (as
described in \cite{Barbosa03}).
Actually, our aim is to characterise
hybrid systems as coalgebras with a discrete (internal) behaviour, and
(external) continuous evolutions.

The cornerstone of this move from continuous to hybrid components is
the notion of tensorial strength for monad $\MH$: a natural
transformation $\tau : Id \times \MH \rightarrow \MH(Id \times Id)$
that commutes with the monad operations and with specific monoidal
structure of the base category (see the formal definition in
\cite{kock:1972}). Indeed, tensorial strength allows us to transport
such systems to $\klH$, via composition:
\begin{center}
  $\infer{\tau \comp c : S \times I \rightarrow \MH (S \times O)}
  {c : S \times I \rightarrow S \times \MH O}$
\end{center}

\begin{Mdefinition}
  Given topological spaces $X,Y \in |\topo|$ a $($right$)$ tensorial
  strength of monad $\MH$ is the function $\tau_{X,Y} : X \times \MH Y
  \rightarrow \MH (X \times Y)$ defined by
  \begin{flalign*}
    & \tau_{X,Y} \> (x,(f,d)) \> \eqdef  ( \pv{\const{x},f}, d).
  \end{flalign*}
\end{Mdefinition}

\noindent
Interestingly, function $\tau$ corresponds to the \emph{uniform characterisation} of
tensorial strength for monads over $\mathbf{Set}$ (\emph{cf.}
\cite{introcoalg}). This entails that all diagrams that need to commute
do commute, and therefore we just need to show that
$\tau$ is continuous. For this, observe that $\tau$ can alternatively be defined  as 
$\pv{\lambda \tau_a, \tau_b} : X \times \MH Y \rightarrow \MH (X \times Y)$ where,
\begin{flalign*}
  & \tau_a \> ((x,(f,d)),t) \> \eqdef (x, f \> t) \\
  & \tau_b \> (x,(f,d)) \> \eqdef d
\end{flalign*}

\noindent
Since $\tau_a,\tau_b$ are continuous, so is $\tau$.

\begin{Mcorollary}
  Natural transformation $\tau : Id \times \MH \rightarrow \MH (Id \times Id)$
  defines a tensorial strength for monad $\MH$.
\end{Mcorollary}


\noindent
Note that one can also define a natural transformation
$\tau_l : \MH \times Id \rightarrow \MH(Id \times Id)$ (known as left
tensorial strength for $\MH$), via the equation
$\tau_l \> \eqdef (\MH sw) \comp \tau \comp sw$.
Moreover, a monad is commutative, if the equation below holds.
\begin{flalign*}
  \tau \kcomp \tau_l = \tau_l \kcomp \tau
\end{flalign*}
\noindent
This is not, however, the case for monad $\MH$, as the following
counter-example reports.

\begin{Mexample}
  Recall the two signal generators, introduced in Example \ref{exam_com}.
  \begin{flalign*}
  & c_1 \> x \> = \>  ( \> x + ( \sin \> \_ \> ), \> 20 \> ), \> \> 
  c_2 \> x \> = \>  ( \> x + \sin \> (3 \times \_ \> ) , \> 20 \> )
  \end{flalign*}
  The application of left and right tensorial strength to the composed
  function $\pv{c_1,c_2} : \mathbb{R} \rightarrow \MH \mathbb{R} \times
  \MH \mathbb{R}$ yields the behaviours depicted below.

  \pgfplotsset{samples=100}
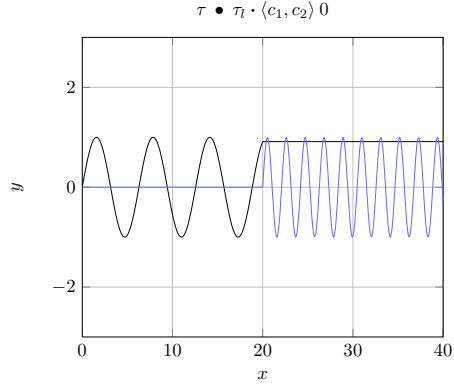
  \begin{center}
    \scalebox{0.7}{  
      \begin{tikzpicture} 
  \begin{axis}[
      title={$\tau \kcomp \tau_l \comp
        \pv{c_1,c_2} \> 0$},
      xlabel=$x$, ylabel=$y$,
      ymin= -3,
      ymax = 3,
      xmin = 0,
      xmax = 40,
      grid = major
  ]
      \addplot[smooth,
      domain = 0:20] { sin(deg(x)) };
      \addplot[smooth,
      domain = 20:40] {sin(deg(20)) };
      \addplot[smooth, blue!60!white,
      domain = 0:20] { 0 };
      \addplot[smooth, blue!60!white,
      domain = 20:40] { sin(3*(deg(x - 20))) };

  \end{axis}
  \end{tikzpicture}
  }  
 \end{center}

  \begin{center}
 \scalebox{0.7}
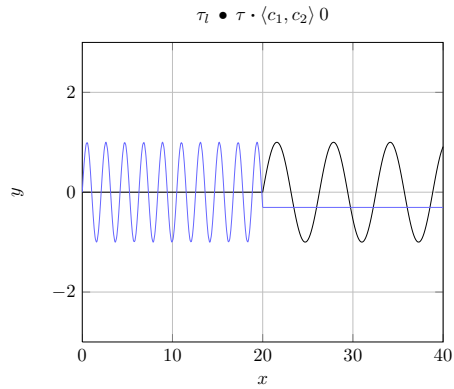
{   
   \begin{tikzpicture} 
    \begin{axis}[
        title={$\tau_l \kcomp \tau \comp
          \pv{c_1,c_2} \> 0$},
        xlabel=$x$, ylabel=$y$,
        ymin= -3,
        ymax = 3,
        xmin = 0,
        xmax = 40,
        grid = major
    ]
        \addplot[smooth,
        domain = 0:20] { 0 };
        \addplot[smooth,
        domain = 20:40] { sin(deg(x - 20)) };
        \addplot[smooth, blue!60!white,
        domain = 20:40] { sin(3*(deg(20))) };
        \addplot[smooth, blue!60!white,
        domain = 0:20] { sin(3*(deg(x))) };

    \end{axis}
    \end{tikzpicture}
  } 
\end{center}

\end{Mexample}

\noindent
Clearly, $\tau \kcomp \tau_l \not = \tau_l \kcomp \tau$; but note that
the plots illustrate an interesting aspect: specification $\tau \kcomp
\tau_l \comp \pv{c_1,c_2}$ reads `first let the component in the left
to act, then the one in the right'; and conversely for $\tau_l \kcomp
\tau \comp \pv{c_1,c_2}$. Moreover, note that each component `waits'
for the other by stalling the corresponding evolution.  This
introduces yet another synchronisation mechanism.

\medskip
\noindent
Equipped with tensorial strength $\tau$, we may now explore two
classical examples of hybrid systems from a component-based
perspective. We start with the \emph{bouncing ball} system.
 
\begin{Mexample}
  Consider a bouncing ball dropped at some positive height and with
  no initial velocity. Due to the gravitational effect, it will
  fall into the ground but then bounce back up,  losing, of course, part
  of its kinetic energy in the process.
\end{Mexample}

\noindent
From this description, one may regard the bouncing ball as a hybrid
component whose (continuous) observable behaviour is the evolution of
its spacial position, whereas the internal memory records velocity,
updated at each bounce. To define such a component we resort to
Newton's equations of motion.
  \begin{flalign*}
    & pos_a \> (v, p, t) = p + v t - \tfrac{1}{2} a t^2, 
    \> \> \>
    vel_a \> (v, t) = v - at
  \end{flalign*}
  from which we can derive the function that, given a positive
  height and a current velocity, returns the time needed to reach the
  ground; formally,
  \begin{flalign*}
    & zpos_a \> (v, p) = \tfrac{\sqrt{2a p + v^2} + v}{a}
  \end{flalign*}
  Let us then define the discrete behaviour of the bouncing ball 
  $b_d : V \times P \rightarrow V$ 
  \begin{flalign*}
    b_d \> (v,p) \> \eqdef vel_g(v, zpos_g(v,p)) \times - 0.5 
  \end{flalign*}
  \noindent
  where $0.5$ is the dampening coefficient. For the continuous
  part $b_c : V \times P \rightarrow \MH P$
  \begin{flalign*}
    b_c \> \eqdef  \pv{pos_g, zpos_g }
  \end{flalign*}
  \noindent
  where $g = 9.8$ (Earth's gravity). The resulting system is a ball
  bouncing on planet Earth, denoted by $b$ and formally defined as $b \>
  \eqdef \tau \comp \pv{b_d,b_c}$.  Assume that the initial state of $b$ is 
  $0$. Then, through the iteration operator, and assuming five as the
  initial position one gets, for instance, the following behaviour.
   \begin{center}
    \scalebox{0.65}
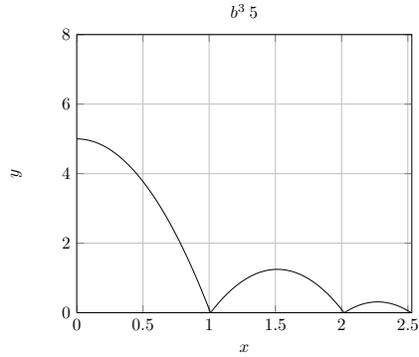
{
    \begin{tikzpicture} 
    \begin{axis}[
        title={$b^3 \> 5$},
        xlabel=$x$, ylabel=$y$,
        ymin= 0,
        ymax = 8,
        xmin = 0,
        xmax = 2.53,
        grid = major
    ]
        \addplot [smooth, 
        domain = 0:2] { 5 - (1/2)*9.8*x^2 };
        \addplot[smooth, shift = {(101.0,0.0)},
        domain = 0:2] { 4.945*x - (1/2)*9.8*x^2 };
        \addplot[smooth, shift = {(202.0,0.0)},
        domain = 0:2] { 2.473*x - (1/2)*9.8*x^2 };
    \end{axis}
    \end{tikzpicture}
    }
  \end{center}

  \noindent
  Analogously, we can define a ball
  bouncing in the Moon (here denoted by letter $c$), and compare the
  behaviour of both bouncing balls by putting them in parallel, with
  the same initial state $0$.
  \begin{center}
    \scalebox{0.65}
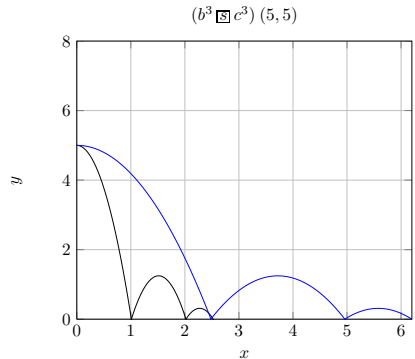
{
  \begin{tikzpicture} 
  \begin{axis}[
      title={$(b^3 \> \boxed{s} \> c^3) \> (5,5)$},
      xlabel=$x$, ylabel=$y$,
      ymin= 0,
      ymax = 8,
      xmin = 0,
      xmax = 6.2,
      grid = major
   ]

      \addplot [smooth, 
      domain = 0:2] { 5 - (1/2)*9.8*x^2 };
      \addplot[smooth, shift = {(101.0,0.0)},
      domain = 0:2] { 4.945*x - (1/2)*9.8*x^2 };
      \addplot[smooth, shift = {(202.0,0.0)},
      domain = 0:2] { 2.473*x - (1/2)*9.8*x^2 };

      \addplot [smooth, color = blue,
      domain = 0:3] { 5 - (1/2)*1.622*x^2 };

      \addplot [smooth, color = blue, 
      domain = 0:3, shift = {(248.0,0.0)} ] 
        { 0 + 2.011*x - (1/2)*1.622*x^2 };

     \addplot [smooth, color = blue, 
      domain = 0:3, shift = {(495.9,0.0)} ] 
        { 0 + 1.0049*x - (1/2)*1.622*x^2 };

  \end{axis}
  \end{tikzpicture}
  }
  \end{center}

\noindent
Note that $(b^3 \> \boxed{s} \> c^3) \not = (b \> \boxed{s} \> c)^3$.
An interesting question to pose is about the durations that components
$\nu b$ and $\nu c$ output. Indeed, the intuition is that durations
are always infinite (since feedback involves infinite sums), however,
due to the Zeno effect, the durations that concern this example are
actually finite: they correspond to the time at which the ball stops
moving. Such durations are given precisely by the computation of
$\pi_2 \comp (\nu b)$ and $\pi_2 \comp (\nu c)$ with respect to a
given input.

\begin{Mexample}
  \emph{Alternating pumping systems} are often used to regulate the water
  level of reservoirs. Consider one that fills two tanks
  alternatively in cycles of ten seconds, which means that some sort
  of internal memory is required $($to remember which was the last tank served$)$.
\end{Mexample}

\noindent
Thus, the discrete part $w_d : S \times L \rightarrow S$ is defined as
\begin{flalign*}
  & w_d \> \eqdef flip \comp \pi_1
\end{flalign*}
\noindent
where $S = \{ \top, \bot \}$ is the discrete state space and $flip$
the function that switches between the elements. Let us assume that the initial
state is $\top$. Then, we define the 
continuous behaviour $w_c : S \times L \rightarrow \MH L$
\begin{flalign*}
  & w_c (s,(l_1,l_2)) \> \eqdef (f_s (l_1,l_2), 10)
\end{flalign*}
\noindent
where $f_{\top} (l_1,l_2) \> \eqdef ((l_1 + \; \_ \;), \; l_2)$ and
$f_{\bot} (l_1,l_2) \> \eqdef (l_1, \; (l_2 + \; \_ \;))$.  As
expected, the pumping system is given by equation $w = \tau \comp
\pv{w_d,w_c}$, which, for input $(0,0)$, yields the following plot.

  \begin{center}
    \scalebox{0.65}
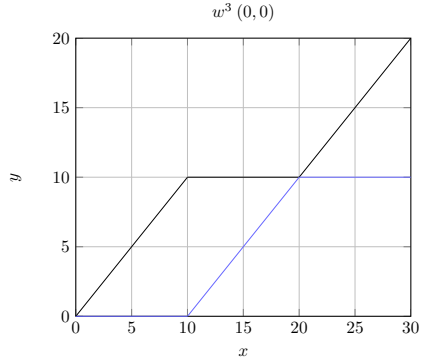
{
  \begin{tikzpicture} 
  \begin{axis}[
      title={$w^3 \> (0,0)$},
      xlabel=$x$, ylabel=$y$,
      ymin= 0,
      ymax = 20,
      xmin = 0,
      xmax = 30,
      grid = major
      ]
      \addplot [smooth, 
      domain = 0:10] { x };
      \addplot [smooth, shift = {(100.0,0)},
      domain = 0:10] { 10 };
      \addplot [smooth, shift = {(200.0,0)},
      domain = 0:10] { 10 + x };
      \addplot [smooth, color= blue!60!white,
      domain = 0:10] { 0 };
      \addplot [smooth, color = blue!60!white, shift = {(100.0,0)},
      domain = 0:10] { x };
      \addplot [smooth, color = blue!60!white, shift = {(200.0,0)},
      domain = 0:10] { 10  };
  \end{axis}
  \end{tikzpicture}
  }
 \end{center}

\noindent
On a different note, it is natural to consider that the pump takes some
time to switch from one tank to the other: for illustration purposes
let us assume that time to be ten seconds. To simulate such a delay we
can define a variant of $copy$, denoted by $copy_{10}$, that always
outputs evolutions with duration ten.  Then, again for input $(0,0)$,
system $(copy_{10} \kcomp w)^3$ outputs

 \begin{center}
   \scalebox{0.65}
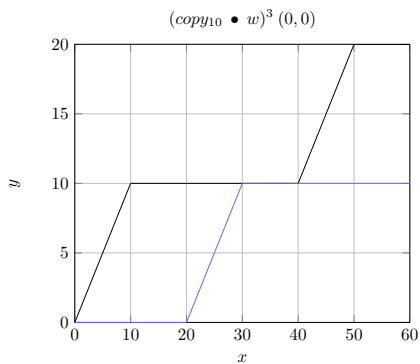
{
  \begin{tikzpicture} 
  \begin{axis}[
      title={$(copy_{10} \kcomp w)^3 \> (0,0)$},
      xlabel=$x$, ylabel=$y$,
      ymin= 0,
      ymax = 20,
      xmin = 0,
      xmax = 60,
      grid = major
      ]
      \addplot [smooth, 
      domain = 0:10] { x };
      \addplot [smooth, shift = {(100.0,0)},
      domain = 0:10] { 10 };
      \addplot [smooth, shift = {(200.0,0)},
      domain = 0:10] { 10 };
      \addplot [smooth, shift = {(300.0,0)},
      domain = 0:10] { 10 };
      \addplot [smooth, shift = {(400.0,0)},
      domain = 0:10] { 10 + x };
      \addplot [smooth, shift = {(500.0,0)},
      domain = 0:10] { 20 };

      \addplot [smooth, color= blue!60!white,
      domain = 0:10] { 0 };
      \addplot [smooth, color = blue!60!white, shift = {(100.0,0)},
      domain = 0:10] { 0 };
      \addplot [smooth, color = blue!60!white, shift = {(200.0,0)},
      domain = 0:10] { x  };
      \addplot [smooth, color = blue!60!white, shift = {(300.0,0)},
      domain = 0:10] { 10  };
      \addplot [smooth, color = blue!60!white, shift = {(400.0,0)},
      domain = 0:10] { 10  };
      \addplot [smooth, color = blue!60!white, shift = {(500.0,0)},
      domain = 0:10] { 10  };
  \end{axis}
  \end{tikzpicture}
  }
 \end{center}

 \noindent
 It is also important to analyse situations in which water flows out. Thus,
 consider a hybrid system $z : 1 \times 1 \rightarrow \MH (1 \times L)$ 
 (with trivial state space  $1$)
 whose continuous part
 \begin{flalign*}
   z_c (\star,\star) \> \eqdef ( \pv{/_2,/_2}, 10)
 \end{flalign*}
 \noindent
 dictates the rate of water flowing out in each tank, here represented by a clock that
 runs at half the normal speed. Then, we specify
 the result of $w$ and $z$ acting together in the same set of
 variables.  For this, we define function $h : (S \times L) \times (1
 \times L) \rightarrow (S \times L) \times (1 \times 1)$ where
 \begin{flalign*}
   h \> ((s,l_1,l_2), (\star,x,y)) \> 
   \eqdef ((s,l_1 \circleddash x, l_2 \circleddash y), (\star,\star))
 \end{flalign*}
 
 \noindent
 Intuitively, function $h$ subtracts water in accordance with the rate
 specified by component $z$.  For input $(0,0)$, system
 $ \corn{ h } \kcomp (w \boxtimes z) $ yields the plot below.

  \begin{center}
    \scalebox{0.65}
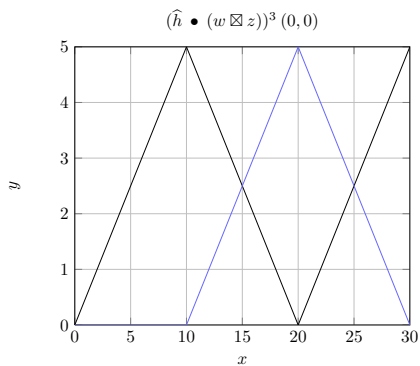
{
  \begin{tikzpicture} 
  \begin{axis}[
      title={$(\corn{ h } \kcomp (w \boxtimes z))^3 \> (0,0)$},
      xlabel=$x$, ylabel=$y$,
      ymin= 0,
      ymax = 5,
      xmin = 0,
      xmax = 30,
      grid = major
      ]
      \addplot [smooth, 
      domain = 0:10] { x - (x/2) };
      \addplot [smooth, shift = {(100.0,0)},
      domain = 0:10] { 5 - (x/2) };
      \addplot [smooth, shift = {(200.0,0)},
      domain = 0:10] { x - (x/2) };
      \addplot [smooth, color= blue!60!white,
      domain = 0:10] { 0 };
      \addplot [smooth, color = blue!60!white, shift = {(100.0,0)},
      domain = 0:10] { x - (x/2) };
      \addplot [smooth, color = blue!60!white, shift = {(200.0,0)},
      domain = 0:10] { 5 - (x/2)  };
  \end{axis}
  \end{tikzpicture}
  }
  \end{center}

\section{Conclusions and future work}
\label{sec:con}

\noindent
It is well known that software systems are becoming prevalently
intertwined with (continuous) physical processes.  Such an
architecture, however, renders their rigorous design (and analysis) a
difficult challenge that calls for a wide, uniform framework combining
the continuous and discrete sides of Mathematics.

As a first step towards a component-based framework for hybrid
systems, in the spirit of \cite{Barbosa03}, this paper showed how
continuous evolutions can be encoded in the form of a strong
(topological) monad. As discussed in Section 1, to capture specific
behavioural models through monads has been a successful path in
Computer Science: such was the case of nondeterministic behaviour,
and the (discrete) probabilistic one; but occurrences in the
continuous domain also exist. A prime example is the \emph{Giry} monad
\cite{prakash98}, which captures stochastic processes and has been
object of study in a number of papers (\emph{e.g.},
\cite{Doberkat:2009,doberkat07,prakash06,jacobs13}).  Along similar
lines, monad $\MH$ provides a categorial universe for continuous, and
hybrid systems, where the effects of continuity over different forms
of composition can be isolated and suitably studied.

This universe, \ie the Kleisli category $\klH$, offers different forms
of system composition, wiring mechanisms, and synchronisation
techniques. For example, Kleisli composition lets the control of an
evolution to be transferred from one system to the other, but also
allows evolutions to be dynamically modified (as observed in the case
of signal amplification). Such behavioural patterns, as discussed in
Section \ref{sec:mon}, are often found in systems like thermostats,
cruise control systems, and signal generators. But more generally, in
control loop systems -- traditionally comprised of a network of
digital controllers that manage a physical process over time through a
feedback loop architecture. In this case, the controllers, possessing
different functionalities, periodically pass control of the physical
process among themselves.

The underlying categorial framework hinted at
several composition operators (through corresponding universal
constructions), and facilitated the elicitation of several
compositional laws. Throughout the paper, the results achieved were
illustrated with classic examples of hybrid systems, namely a
thermostat, a bouncing ball, and a water tank system.

\subsection{Related work}

\emph{Hybrid automata} \cite{hybridautomata} are the \emph{de facto}
formalism for the specification of hybrid systems.  Roughly speaking,
they are a variant of classic automata that allows variables to
continuously evolve while in a state. This defines the continuous
behaviour of an hybrid system, which is then paired with discrete
actions given by the usual state transitions.  Parallel composition of
hybrid automata proceeds similarly to the classic case, where common
labels act as synchronising events.  Interestingly, in
\cite{sifakis98} Bornot and Sifakis introduced additional
synchronisation mechanisms that make one system wait for the evolution
of the other to end, or, on the contrary, force it to finish
earlier. This seems to be intimately related to whatever monoidal
structure is given to functor $\MH$.

During the last years there were also developments concerning the
addition of new dimensions to hybrid automata: for example,
\cite{pnueli10} shows how to take \emph{reaction times} into
consideration in a compositional setting. In our case, 
we took advantage of dawdler components, like $copy_{10}$, to
introduce such delays.

The `rationale' underlying hybrid automata is powerful, and highly
intuitive, but in some cases lacks expressive power: for example,
those systems in which evolutions can be dynamically changed
by some of the components are very hard to specify.  Moreover, aside
from parallel composition, the authors have no knowledge of deep
developments that concern new compositional operators for hybrid
automata.

The industrial tool \textsc{Simulink}\footnote{
  \url{http://www.mathworks.com/products/simulink} }, on the other
hand, offers a highly expressive component-based language, and is thus
closely related to the framework proposed in this paper. Indeed,
\textsc{Simulink} supports a rich palette of compositional operators,
and computational units. It possesses behavioural patterns that
involve dynamical alteration of evolutions, delays, and
synchronisation. Moreover, the transfer of the control of some
evolution is not hard to define. All this renders \textsc{Simulink} a
very interesting tool.  The cost is the lack of a clear semantics, which
impairs formal analysis and the elicitation of compositional laws --
actually, some recent efforts have been made towards the formal
verification of \textsc{Simulink} models in alternative tools
(\emph{cf.} \cite{zou15,robert14}).  In addition, the components
available are rather limited in what concerns the characterisation of
their internal memory and respective transition dynamics.

It would be interesting to study the embedding of (a subset of)
\textsc{Simulink}'s language into $\klH$. In principle, $\klH$ could
act as a tool complement, providing a basis for the formal analysis of
(critical fragments) of hybrid systems.  We stress, however, that we
do not aim at emulating \textsc{Simulink}, but rather at a suitable
coalgebraic framework for hybrid components, where we consider the
discrete transitions to be internal behaviour, and the continuous
evolutions the observable part. From this point of view,
\textsc{Simulink} is very distant from such a line of work.

There is also a close relation between the work here reported and
P. H{\"o}fner's algebra of hybrid systems \cite{hofner_phdthesis}:
the latter's main operator is used to concatenate
evolutions. Moreover, the algebra possesses secondary operators, like
parallelism and synchronisation, that are equally available in
$\klH$. Our approach, however, and differently from
P. H{\"o}fner's calculus, is structured around a monad that encodes
the notion of continuous evolution; this brings up a number of canonical
constructions and smooths the integration with other behavioural
effects, such as nondeterminism or probabilistic behaviour.

Finally, a few categorial models for hybrid systems have been proposed
along the last two decades. For example, document \cite{sernadas00}
introduces an institution -- in essence, a categorial rendering of a
logic -- for hybrid systems, and provides basic forms of composition
such as free aggregation (\ie, parallelism without interaction) and
interconnection where some attributes and events are shared between
two systems. Around the same time, Jacobs \cite{Jacobs2000} suggested
an object oriented coalgebraic framework where hybrid systems are
regarded as coalgebras equipped with a monoid action: coalgebras
define the discrete transitions, and monoid actions the continuous
evolutions.  Some years later Haghverdi \emph{et. al}
\cite{haghverdi05} explored the connection between a formalisation of
hybrid systems (close to hybrid automata) and open maps.  The
objective was to provide appropriate notions of bisimulation both for
dynamical, and hybrid systems.  Composition mechanisms, however, were
not studied in this context.

\subsection{Future work.}

\noindent
Our next step is the development of a calculus of hybrid components
(as in \cite{Barbosa03}) based on monad $\MH$ and its Kleisli
category. The calculus from \cite{Barbosa03}, in its coalgebraic
spirit, is bisimulation-based, with bisimulation given as the usual
span of simulations \cite{introcoalg}. The framework that this paper
sets, however, offers a promising basis to explore alternative notions
of (bi)simulation for continuous and hybrid systems.  This has points
of contact with the work of Haghverdi \emph{et. al} in
\cite{haghverdi05}; but note that we use coalgebraic machinery, and
follow a component-based perspective, which makes possible to study
the relation between (bi)simulation and (the different) compositional
operators.

A second line of research concerns the development of a taxonomy of
continuous, and hybrid systems living in $\klH$.  Indeed, as Stauner
showed at the beginning of the century in his PhD thesis
\cite{Stauner_01}, topologies are useful to elicit a number of
important properties.  For example, the notion of \emph{robustness}
(prevalent in control theory) becomes simple to formulate:
intuitively, a system is robust if small changes in the input lead to
very similar evolutions. In $\klH$, since each system has a
topological semantic base, one can express how robust it is by varying
the topology in its source object. At one limit, if the topology is
discrete, the system is seen as \emph{chaotic}. At the other end, \ie, 
if the topology is indiscrete, the system must always output the same
evolution.

Actually, the compositional nature that underlies $\klH$ allows us to
reason about the robustness of the system at hands through the
analysis of (the robustness of) its simpler constituents.  One
disadvantage of this approach is that composition in $\klH$ is
\emph{strict}, in the sense that components with different topologies
in the connecting points cannot be composed. For example, it is hard
to put a chaotic component after a robust one. Part of our current
research tries to relax this condition while maintaining stability,
whenever possible.

\subparagraph*{Acknowledgements} 
This work is financed by the ERDF –
European Regional Development Fund through the Operational Programme
for Competitiveness and Internationalisation - COMPETE 2020 Programme
and by National Funds through the Portuguese funding agency, FCT -
Funda\c{c}\~{a}o para a Ci\^{e}ncia e a Tecnologia within project
POCI-01-0145-FEDER-016692.  

The first author is also sponsored by FCT
grant SFRH/BD/52234/2013, and the second one by FCT grant
SFRH/BSAB/113890/2015. 

We are grateful for the interesting discussions
that the first author had with Ichiro Hasuo, Toshiki Kataoka, and
Soichiro Fujii about the characterisation of monad $\MH$.  Finally, we
gratefully acknowledge the anonymous reviewers for their interesting
comments along the revision process.

\bibliographystyle{elsarticle-num} 
\bibliography{biblio}

\section*{Appendix}

\noindent {\bf Proof of Lemma \ref{lem:mult_nat}}.  
The proof is
divided in two parts: the first establishes continuity of the
mappings, the second concerns naturality. Consider the mapping $\mu_X
: \MH \MH X \rightarrow \MH X$; we are going to show its
continuity. First we observe that $\mu_X$ can be alternatively defined
as $\pv{\lambda a, b}$ where
  \begin{flalign*}
    & a \> \eqdef \MH \MH X \times \Rz \typ{i \times id} (X^{\Rz
      \times \Rz} \times \Dur)
    \times \Rz \typ{conc} X   \\
    & i \> \eqdef \MH \MH X \typ{\pi_1^{\Rz} \times id}
    (X^{\Rz})^{\Rz} \times \Dur \typ{\cong} X^{\Rz \times \Rz} \times \Dur
  \end{flalign*}
  \noindent
  for $conc \> ((f,d),t) \> \eqdef f \> (t \curlywedge d, t \circleddash d)$. The
  definitions clearly show that $a$ is continuous. For function $b$ we have
  \begin{flalign*}
    & b \> \eqdef \MH \MH X \typ{\pi_2^{\Rz} \times id} \Dur^{\Rz} \times \Dur
    \typ{c} \Dur 
  \end{flalign*}
  \noindent
  where
  $c \> (f,d) \> \eqdef \hmul{(f \> d) + d}{(d \not =
    \infty)}{\infty}$.
  Since the canonical restriction
  $(+) \comp \pv{ev,\pi_2} : \Dur^{\Rz} \times \Rz \rightarrow \Dur$
  of $c$ is continuous we just need to show that the latter is
  continuous at infinity. Actually, this comes for free once proved
  that given any neighbourhood $N \supseteq (x, \infty]$ in $\Dur$ of
  $\infty$ we can find a neighbourhood $V$ in $\Dur^{\Rz} \times \Dur$
  of $(f,\infty)$ such that $c \> (V) \subseteq N$.

  Consider neighbourhood $\Dur^{\Rz} \times (x,
  \infty]$. It is clear that $c \> (\Dur^{\Rz} \times (x,\infty])
  \subseteq (x,\infty] \subseteq N$.
  
  \medskip
  \noindent
  Next we show that $\mu$ is natural, \ie, that for any continuous
  function $h : X \rightarrow Y$ the diagram below commutes. 

   \begin{center}
   \begin{tabular}{c}
     \xymatrix {
     \MH \MH X \ar[r]^{\MH \MH h} \ar[d]_{\mu_X} & \MH \MH Y \ar[d]^{\mu_Y} \\
     \MH X \ar[r]_{\MH h } & \MH Y
     }
   \end{tabular}
 \end{center}

 \noindent
 First we assume that $(f,d) \in \mathcal{H} \mathcal{H} X$ has
 finite duration,
  \begin{eqnarray*}
    && \mu \comp \MH \MH h \> (f,d) 
       \justl{=}{ Definition of $\MH$, $\mu$}
       ( \theta \comp \MH h \comp f, d) \conc (\MH h \comp f \> d)
       \justl{=}{$\theta$ is natural}
       (h \comp \theta \comp f, d) \conc (\MH h \comp f \> d)
       \justl{=}{Definition of $\MH$, composition }
       \MH h \> (\theta \comp f, d) \conc \MH h \> (f \> d)
       \justl{=}{ ($\conc$) is natural}
       \MH h \> ( (\theta \comp f , d) \conc (f \> d) )
       \justl{=}{ Definition of $\MH$, $\mu$ }
       \MH h \comp \mu \> (f,d)
  \end{eqnarray*}

\noindent
The proof for the case in which $(f,d) \in \mathcal{H} \mathcal{H} X$
has infinite duration is analogous to the above.  

\qed

\noindent
{\bf Proof of Theorem \ref{theo:monad}}.
  We have to show that the following diagrams commute.
   \begin{center}
     \begin{tabular}{c c c}
         \xymatrix{
         \MH \ar[r]^{\eta_{\MH}}
         \ar[dr]_{1_{\MH}} & \MH^2  \ar[d]_\mu & 
         \MH \ar[l]_{\MH \eta} 
             \ar[dl]^{1_{\MH}} \\
         & \MH & 
         } 
          & & 
          \xymatrix{
           \MH^3 \ar[r]^{\mu_{\MH}} 
           \ar[d]_{\MH {\mu}} & 
           \MH^2 \ar[d]^{{\mu}} \\
           \MH^2 \ar[r]_\mu & \MH
           } 
       \end{tabular}
  \end{center}
  Note that the proof below becomes much more simpler if the evolutions
  involved have infinite duration.

  Let us start with the left triangle.

  \begin{eqnarray*}
    && \mu \comp \eta \> (f,d)
    \justl{=}{Definition of $\eta$}
    \mu \> ( \underline{(f,d)}, 0)
    \justl{=}{Definition of $\mu$}
    (\theta \comp \underline{(f,d)},0) \conc
    (\underline{(f,d)} \> 0)
    \justl{=}{Definition of constant}
    (\theta \comp \underline{(f,d)},0) \conc
    (f,d)
    \justl{=}{ Definition of $\conc$, definition of constant }
    (f,d)
  \end{eqnarray*}

  \noindent
  For the right triangle we have,
  \begin{eqnarray*}
    && \mu \comp \MH \eta \> (f,d)
    \justl{=}{ Definition of $\MH$}
    \mu \> (\eta \comp f, d)
    \justl{=}{ Definition of $\mu$}
    (\theta \comp \eta \comp f, d) \conc
    (\eta \comp f \> \> d)
    \justl{=}{ Definition of $\eta$ }
    (\theta \comp \eta \comp f, d) \conc
    (\underline{f \> d},  0)  
    \justl{=}{ Definition of $\conc$ }
    (\theta \comp \eta \comp f, d) 
    \justl{=}{ Eilenberg-Moore }
    (f,d)
  \end{eqnarray*}

  \noindent
  It remains to show that the square commutes. Before giving the
  formal proof, we present the corresponding intuition from a
  geometric perspective. 

  Let us then start
  by observing that an element in $\MH \MH X$, 
  may be intuitively seen as a square, where each column is
  a function in $\MH X$. Then, note that multiplication ($\mu : 
  \MH \MH X \rightarrow \MH X$) 
  keeps just the first row and
  last column of the square, as illustrated below.

    \begin{center}
      \scalebox{0.8}
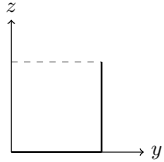
{
        \begin{tikzpicture} 
           \draw[color=gray, dashed] (0,1.5) -- (1.5,1.5);
           \draw[->] (0,0) -- (2.2,0) node[right] {$y$};
           \draw[->] (0,0) -- (0,2.2) node[above] {$z$};
           \draw[-, thick] (0,0) -- (1.5,0);
           \draw[-, thick] (1.5,0) -- (1.5,1.5);
       \end{tikzpicture}
       }
   \end{center}

  \noindent
  As expected, the intuitive picture of an element in $\MH^3 X$
  is a cube,

    \begin{center}
      \scalebox{0.8}
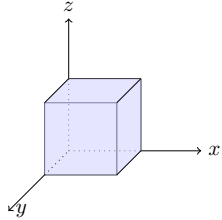
{
        \begin{tikzpicture} 

           \draw[->] (1.2,0) -- (2.2,0) node[right] {$x$};
           \draw[->] (0,1.2) -- (0,2.2) node[above] {$z$};
           \draw[->] (-0.4,-0.4) -- (-1,-1) node [right] {$y$};

           \draw[thin,dotted] (0,0) -- (1.2, 0);
           \draw[thin,dotted] (0,0) -- (0, 1.2);
           \draw[thin] (0,1.2) -- (1.2, 1.2);
           \draw[thin] (1.2,0) -- (1.2, 1.2);
           \draw[thin] (-0.4,-0.4) rectangle (0.8,0.8);

           \fill [fill=blue!20!white, opacity = 0.5] (-0.4,-0.4) 
           rectangle (0.8,0.8);
           
           \path [fill=blue!20!white, opacity = 0.5] (0.8, -0.4) -- (1.2,0)
           -- (1.2, 1.2) -- (0.8, 0.8);

           \path [fill=blue!20!white, opacity = 0.5] (-0.4, 0.8) -- (0, 1.2)
           -- (1.2, 1.2) -- (0.8,0.8);

           \draw[thin] (0.8, -0.4) -- (1.2,0);
           \draw[thin] (0.8, 0.8) -- (1.2,1.2);

           \draw[thin,dotted] (-0.4,-0.4) -- (0,0);  
           \draw[thin] (-0.4, 0.8) -- (0, 1.2);

       \end{tikzpicture}
       }
   \end{center}

   \noindent
   such that a projection on the $x$-axis yields an element of
   $\MH \MH X$ (geometrically, a square as described above).

   Let us now observe that, resorting to multiplication, we can reduce
   the cube into a square.  Actually, we can do this in two different
   ways: via $\mu : \MH^3 X \rightarrow \MH \MH X$, or
   $\MH \mu : \MH^3 X \rightarrow \MH \MH X$. In the former case, only
   the front and right surfaces are kept (picture below in the
   left). In contrast, function $\MH \mu$ applies $\mu$ to
   each projection on the $x$-axis, and thus only the bottom and back
   surfaces are kept (picture below in the right).

   \begin{multicols}{2}
      
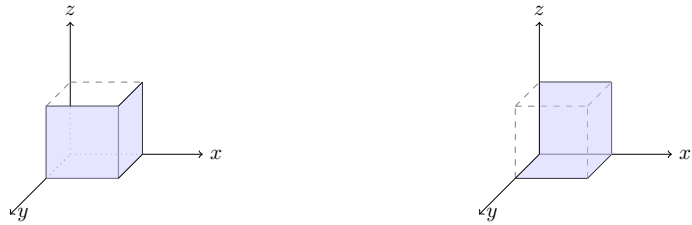
     \begin{center}
      \scalebox{0.8}{
        \begin{tikzpicture} 

           \draw[->] (1.2,0) -- (2.2,0) node[right] {$x$};
           \draw[->] (0,0.8) -- (0,2.2) node[above] {$z$};
           \draw[->] (-0.4,-0.4) -- (-1,-1) node [right] {$y$};

           \draw[color = gray, thin, dotted] (0,0) -- (1.2, 0);
           \draw[color = gray, thin, dotted] (0,0) -- (0, 0.8);
           \draw[color = gray, thin, dashed] (0,1.2) -- (1.2, 1.2);
           \draw[thin] (1.2,0) -- (1.2, 1.2);
           \draw[color = black, thin] (-0.4,-0.4) rectangle (0.8,0.8);

           \fill [fill=blue!20!white, opacity = 0.5] (-0.4,-0.4) 
           rectangle (0.8,0.8);

           \path [fill=blue!20!white, opacity = 0.5] (0.8, -0.4) --
           (1.2,0) -- (1.2, 1.2) -- (0.8, 0.8);

           \draw[thin] (0.8, -0.4) -- (1.2,0);
           \draw[thin] (0.8, 0.8) -- (1.2,1.2);

           \draw[color = gray, thin,dotted] (-0.4,-0.4) -- (0,0);  
           \draw[color = gray, thin,dashed] (-0.4, 0.8) -- (0, 1.2);

       \end{tikzpicture}
       }
     \end{center}
     
     \columnbreak

      \begin{center}
      \scalebox{0.8}{
        \begin{tikzpicture} 

           \draw[->] (1.2,0) -- (2.2,0) node[right] {$x$};
           \draw[->] (0,1.2) -- (0,2.2) node[above] {$z$};
           \draw[->] (-0.4,-0.4) -- (-1,-1) node [right] {$y$};

           \draw[thin] (0,0) -- (1.2, 0);
           \draw[thin] (0,1.2) -- (1.2, 1.2);
           \draw[thin] (1.2,0) -- (1.2, 1.2);

           \draw [thin] (-0.4,-0.4) -- (0.8, -0.4);
           \draw [color = gray, thin, dashed] (0.8,-0.4) -- (0.8, 0.8);
           \draw [color = gray, thin, dashed] (-0.4,0.8) -- (0.8, 0.8);
           \draw [color = gray, thin, dashed] (-0.4,-0.4) -- (-0.4, 0.8);
           \draw[thin] (0.8, -0.4) -- (1.2,0);
           \draw[color = gray, thin, dashed] (0.8, 0.8) -- (1.2,1.2);

           \fill [fill=blue!20!white, opacity = 0.5] (0,0) 
           rectangle (1.2,1.2);

           \path [fill=blue!20!white, opacity = 0.5] (-0.4, -0.4) --
           (0.8,-0.4) -- (1.2, 0) -- (0, 0);

           \draw[color = gray, thin,dashed] (-0.4, 0.8) -- (0, 1.2);

            \draw[] (-0.4,-0.4) -- (0,0);  
            \draw[] (0,0) -- (0, 1.2);
       \end{tikzpicture}
       }
   \end{center}

   \end{multicols}

   \noindent
   Finally, applying $\mu : \MH \MH X \rightarrow \MH X$ to the
   resulting squares yields the same result,

         \begin{center}
      \scalebox{0.8}
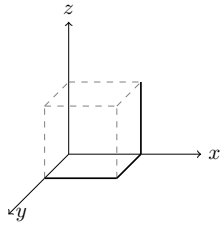
{
        \begin{tikzpicture} 

           \draw[->] (1.2,0) -- (2.2,0) node[right] {$x$};
           \draw[->] (0,1.2) -- (0,2.2) node[above] {$z$};
           \draw[->] (-0.4,-0.4) -- (-1,-1) node [right] {$y$};

           \draw[thin] (0,0) -- (1.2, 0);
           \draw[thin] (0,0) -- (0, 1.2);
           \draw[color = gray, thin, dashed] (0,1.2) -- (1.2, 1.2);
           \draw[thick] (1.2,0) -- (1.2, 1.2);

           \draw [thick] (-0.4,-0.4) -- (0.8, -0.4);
           \draw [color = gray, thin, dashed] (0.8,-0.4) -- (0.8, 0.8);
           \draw [color = gray, thin, dashed] (-0.4,0.8) -- (0.8, 0.8);
           \draw [color = gray, thin, dashed] (-0.4,-0.4) -- (-0.4, 0.8);
           \draw[thick] (0.8, -0.4) -- (1.2,0);
           \draw[color = gray, thin, dashed] (0.8, 0.8) -- (1.2,1.2);

           \draw[thin] (-0.4,-0.4) -- (0,0);  
           \draw[color = gray, thin,dashed] (-0.4, 0.8) -- (0, 1.2);

       \end{tikzpicture}
       }
   \end{center}

  More formally, we reason
  \begin{eqnarray*}
    && 
       \mu \comp \MH \mu \> (f,d)
       \justl{=}{Definition of $\MH$}
       \mu \> (\mu \comp f, d)
       \justl{=}{Definition of $\mu$}
       (\theta \comp \mu \comp f, d)  \conc (\mu \comp f \> d)
       \justl{=}{Eilenberg-Moore}
       (\theta \comp \MH \theta \comp f, d) \conc (\mu \comp f \> d)
       \justl{=}{ Let $f \> d = (f',d')$, definition of $\conc$ }
       \big ( \> (\theta \comp \MH \theta \comp f, d \>)  
       \conc  ( \theta \comp f',d') \> \big ) \conc ( f' \> d')
       \justl{=}{ $\theta$ is natural }
       \big ( \> (\theta \comp \theta \comp f, d \>)  
       \conc  ( \theta \comp f',d') \> \big ) \conc ( f' \> d')
       \justl{=}{ Notation ($f \> d$), definition of $\conc$, definition of $\MH$ }
       \big ( \> \MH \theta \> ( (\theta \comp f, d \>)  
       \conc  (f \> d) ) \> \big ) \conc ( f' \> d')       
       \justl{=}{ Definition of $\mu$}
       \big ( \> \MH \theta \comp \mu \> (f,d)  \> \big ) \conc ( f' \> d')              
       \justl{=}{ Definition of $\mu$ }
       ( \> \MH \theta \comp \mu \> (f,d) \> ) \conc 
       (\> \pi_1 \comp \mu \> (f,d) \> \> \pi_2 \comp \mu \> (f,d) \>)
       \justl{=}{Definition of $\MH$}
       ( \theta \comp \pi_1 \comp \mu \> (f,d), \pi_2 \comp \mu \> (f,d) )
       \conc (\pi_1 \comp \mu \> (f,d) \> \> \pi_2 \comp \mu \> (f,d))
       \justl{=}{Definition of $\mu$ }
       \mu \> ( \mu \> (f,d) )
       \justl{=}{Composition}
       \mu \comp \mu \> (f,d)
  \end{eqnarray*}

\qed

\noindent
{\bf Proof of Theorem \ref{theo_pull}}.
  Consider the following pullback in $\topo$ 

    \begin{center}
      \begin{tabular}{c}
        \xymatrix {
          A \times_C B \ar[r]^(.6){\pi_2} \ar[d]_{\pi_1} \pullbackcorner & B \ar[d]^{g} \\
          A \ar[r]_{f} & C
        }
      \end{tabular}
    \end{center}

  \noindent
  where $f$ and $g$ are arbitrary continuous functions. We need to show 
  that 
      \begin{center}
      \begin{tabular}{c}
        \xymatrix {
          \MH (A \times_C B) \ar[r]^(.6){\MH \pi_2} \ar[d]_{\MH \pi_1} 
          \pullbackcorner & \MH B 
          \ar[d]^{\MH g} \\
          \MH A \ar[r]_{\MH f} & \MH C
        }
      \end{tabular}
    \end{center}
  \noindent
  also forms a pullback in $\topo$. 

  For this, observe that functor $\MH$ comes from the composition of
  functors $( \> \_ \> )^{\Rz}$, and $( \> \_ \> \times \Dur)$, both of which
  preserve pullbacks. Indeed, they give rise to the commuting diagram
  \begin{center}
  \begin{tabular}{c}
    \xymatrix {
      X \ar@{.>}[dr]^{\gamma \comp \langle {c_1,c_2} \rangle }
      \ar@{->}@/_2pc/[ddr]_{c_1}
      \ar@{->}@/^2pc/[drrr]^{c_2}
      & & &  \\
      & (A \times_C B)^{\Rz} \times \Dur \ar@{->}[rr]^{ {\pi_2}^{\Rz} \times id }
      \ar@{->}[d]_{ {\pi_1}^{\Rz} \times id }
      & & B^{\Rz} \times \Dur \ar@{->}[d]^{ {f}^{\Rz} \times id } \\
      & A^{\Rz} \times \Dur \ar@{->}[rr]_{ {g}^{\Rz} \times id } & & C^{\Rz} 
      \times \Dur
    }
  \end{tabular}
  \end{center}
  \noindent
  where $\gamma\> ((e_1,d),(e_2,d)) = (\pv{e_1,e_2},d)$. Let us denote
  $\gamma \comp \pv{c_1,c_2}$ by $\pulb{c_1,c_2}$. 

  Since functor $\MH$ forces specific conditions on
    evolutions (recall that $(e,d) \in \MH X$ implies $e \comp
    \curlywedge_d = e$) some work remains to be done. In fact,
    we need to show that $\img \> \pulb{c_1,c_2} \subseteq \MH (A
    \times_C B)$ whenever $\img \> c_1 \subseteq \MH A$, $\img \> c_2
    \subseteq \MH B$, and $c_1,c_2$ make the outer square to commute.
    In other words, we need to show that, under these conditions,
    $\pulb{c_1,c_2}$ factors through $\iota : \MH (A \times_C B)
    \hookrightarrow (A \times_C B)^{\Rz} \times \Dur$;
    diagrammatically,
  \begin{center}
  \begin{tabular}{c}
    \xymatrix {
       X \ar[r]^(0.3){  \pulb{c_1,c_2} { }  }
         \ar@{.>}[dr] & (A \times_C B)^{\Rz} \times \Dur \\
       & \MH (A \times_C B) \ar@{_{(}->}[u]_{\iota}
     }
  \end{tabular}
  \end{center} 

  \noindent
  Consider an element $x \in X$, and denote $\pulb{c_1,c_2} \> x$ by
  $(\pv{e_1, e_2}, d)$. Since by assumption $e_1 \comp \curlywedge_d = e_1$,
  $e_2 \comp \curlywedge_d = e_2$, it is clear that
  $\pv{e_1,e_2} \comp \curlywedge_d = \pv{e_1,e_2}$ and therefore
  $(\pv{e_1, e_2}, d) \in \MH (A \times_C B)$. 

  \qed

\medskip
\noindent
{\bf Proof of Theorem \ref{teodelta}}. We need to show that the following diagrams
commute.

\begin{center}
  \begin{tabular}{c}
    \xymatrix {
      (\MH X \times \MH Y) \times \MH Z
      \ar[rr]^{\alpha}
      \ar[d]_{\delta \times id}
      && 
      \MH X \times ( \MH Y \times \MH Z)
      \ar[d]^{id \times \delta} \\
      \MH (X \times Y) \times \MH Z
      \ar[d]_{\delta} && 
      \MH X \times \MH (Y \times Z )
      \ar[d]^{\delta} \\
      \MH ((X \times Y) \times Z)
      \ar[rr]_{\MH \alpha }
      && \MH ( X \times  ( Y \times Z ))
    }
  \end{tabular}
  \end{center}

  \begin{center}
  \begin{tabular}{c c c c}
       \xymatrix{
           \MH X \times 1 \ar[rr]^{id \times m} 
           \ar[d]_{\pi_1}
           && \MH X \times \MH 1 \ar[d]^{\delta} \\
           \MH X && \MH (X \times 1) 
           \ar[ll]^{\MH \pi_1}
     } 
          & & & 
     \xymatrix{
           1 \times
           \MH X \ar[rr]^{m \times id}  \ar[d]_{\pi_2} 
           && \MH 1 \times \MH X \ar[d]^{\delta} \\
           \MH X && \ar[ll]^{\MH \pi_2} 
           \MH (1 \times X)
     } 
   \end{tabular}
   \end{center}

  We start with the upper square.

  \begin{eqnarray*}
    && \MH \alpha \comp \delta \comp (\delta \times id) \> \>
    \big ( \> ((e_1,d_1),(e_2,d_2)),(e_3,d_3) \> \big)
    \justl{=}{Definition of $\delta$ and $\MH$ }
    ( \> \alpha \comp \pv{\pv{e_1,e_2},e_3}, ((d_1 \curlyvee d_2) \curlyvee d_3) \>)
    \justl{=}{Definition of product, $\curlyvee$ is associative }
    (\> \pv{e_1,\pv{e_2,e_3}}, (d_1 \curlyvee ( d_2 \curlyvee d_3 )) \>)
    \justl{=}{Definition of $\delta$}
    \delta \> ( \> (e_1,d_1), (\pv{e_2,e_3}, d_2 \curlyvee d_3) \>)
    \justl{=}{Definition of $id \times \delta$ }
    \delta \comp (id \times \delta) \> \> \big ( \> (e_1,d_1), ((e_2,d_2),(e_3,d_3)) \> \big )
    \justl{=}{Definition of $\alpha$}
    \delta \times (id \times \delta) \comp \alpha \> \>
    \big ( \> ((e_1,d_1), (e_2,d_2)),(e_3,d_3) \> \big )
  \end{eqnarray*}

  \medskip
  \noindent
  Then, for the diagram above in the left we reason, and proceed similarly
  with the one in the right.
  \begin{eqnarray*}
    && \MH \pi_1 \comp \delta \comp (id \times m) \> ((f,d), \star)
    \justl{=}{ Definition of $m$, $\delta$, and $\MH$ }
    ( \pi_1 \comp \pv{f, \const{\star}}, d \curlyvee 0 )
    \justl{=}{ Cancellation $\times$, $0$ is the identity element 
      (for $\curlyvee$) }
    (f,d)
    \justl{=}{ Definition of $\pi_1$ }
    \pi_1 \> ((f,d), \star)
  \end{eqnarray*}

\qed


\end{document}